\documentclass{lmcs}
\pdfoutput=1
\usepackage[utf8]{inputenc}

\usepackage{lastpage}
\lmcsdoi{21}{4}{7}
\lmcsheading{}{\pageref{LastPage}}{}{}%
{Oct.~03,~2024}{Oct.~09,~2025}{}


\usepackage{hyperref}
\usepackage{cleveref}
\crefname{thm}{theorem}{theorems}
\Crefname{thm}{Theorem}{Theorems}
\crefname{thmC}{theorem}{theorems}
\Crefname{thmC}{Theorem}{Theorems}
\crefname{lem}{lemma}{lemmas}
\Crefname{lem}{Lemma}{Lemmas}
\crefname{cor}{corollary}{corollaries}
\crefname{cor}{corollary}{corollaries}
\crefname{prop}{proposition}{propositions}
\crefname{prop}{Proposition}{Propositions}
\crefname{defi}{definition}{definitions}
\crefname{defi}{Definition}{Definitions}
\crefname{exa}{example}{examples}
\crefname{exa}{Example}{Examples}
\usepackage{subcaption}
\makeatletter         % include caption macros from class file to preserve lmcs caption style
\def\figurecaption#1#2{\noindent\hangindent 40pt
                       \hbox to 36pt {\small\sl #1 \hfil}
                       \ignorespaces {\small #2}}
\long\def\@makecaption#1#2{
  \vskip 10pt 
  \settowidth{\@tempdima}{#2}
  \ifdim\@tempdima>0pt
       \setbox\@tempboxa\hbox{#1: #2}
     \else
       \setbox\@tempboxa\hbox{#1 #2}
   \fi
   \ifdim \wd\@tempboxa >\hsize               % IF longer than one line:
       \begin{list}{#1:}{
       \settowidth{\labelwidth}{#1:}
       \setlength{\leftmargin}{\labelwidth}
       \addtolength{\leftmargin}{\labelsep}
        }\item #2 \end{list}\par   % Output in quote mode
     \else                                    %   ELSE  center.
       \hbox to\hsize{\hfil\box\@tempboxa\hfil}  
   \fi}
\makeatother

\newcommand{\includefigure}[1]{\includegraphics{#1.pdf}}
\usepackage{amsmath, amssymb} \usepackage{stmaryrd} \usepackage{wasysym}
\usepackage{amsthm}
\usepackage{tikz-cd} \usetikzlibrary{decorations.markings}
\tikzset{kleisli/.style={ postaction={decorate, decoration={markings, mark= at position 0.5 with {
          \draw circle[radius=1.5pt]; }}
    }}}
\tikzset{e1/.style={ postaction={decorate, decoration={markings, mark= at position 1 with {
          \node[below] (tempnode) {\tiny{1}};}}
    }}}
\tikzset{e2/.style={ postaction={decorate, decoration={markings, mark= at position 1 with {
          \node[below] (tempnode) {\tiny{2}};}}
    }}}
\tikzset{m1/.style={ postaction={decorate, decoration={markings, mark= at position 0 with {
          \node[below] (tempnode) {\tiny{1}};}}
    }}}
\tikzset{m2/.style={ postaction={decorate, decoration={markings, mark= at position 0 with {
          \node[below] (tempnode) {\tiny{2}};}}
    }}}
\usepackage{algorithm, algorithmic}

\Crefname{ALC@unique}{Line}{Lines}
\newcounter{myalg}
\AtBeginEnvironment{algorithmic}{\refstepcounter{myalg}}
\makeatletter
\@addtoreset{ALC@unique}{myalg}
\makeatother

\DeclareMathSymbol{:}{\mathpunct}{operators}{"3A}
\DeclareMathOperator{\lgcd}{lgcd} \DeclareMathOperator{\red}{red}
\DeclareMathOperator{\Irr}{Irr} \DeclareMathOperator{\rk}{rk}
\newcommand{\dual}[1]{{#1}^{op}} \newcommand{\inv}[1]{{#1}^{-1}}
\newcommand{\invertibles}[1]{{#1}^{\times}}
\newcommand{\LeftDivide}{\textsc{LeftDivide}}

 \DeclareMathOperator{\id}{id}
\DeclareMathOperator{\Reach}{Reach} \DeclareMathOperator{\Obs}{Obs}
\DeclareMathOperator{\Min}{Min} \DeclareMathOperator{\Total}{Total}
\DeclareMathOperator{\Prefix}{Prefix}
 \DeclareMathOperator{\Id}{Id}
\newcommand{\op}{\mathrm{op}}
\newcommand{\Set}{\mathbf{Set}} \newcommand{\Trans}{\mathbf{Trans}}
 \newcommand{\C}{\mathcal{C}}
\newcommand{\Kl}{\mathbf{Kl}} \newcommand{\Rel}{\mathbf{Rel}}
\newcommand{\Auto}[1]{\mathbf{Auto}(#1)}
\newcommand{\BiAuto}[2]{\mathbf{Auto}_{#1}(#2)}
\newcommand{\outval}[1]{\L({\triangleright} {#1} {\triangleleft})}
\newcommand{\klarrow}[1][->]{%
  \mathrel{\tikz [line width=.11ex, double distance=.33ex]
    \draw[#1, kleisli]
    (0,0) -- (0.7,0);}}
\newcommand{\klarrowin}[2][>->>]{%
  \mathrel{\tikz [baseline=-.5ex, line width=.11ex, double distance=.33ex]
    \draw[#1, kleisli, #2]
    (0,0) -- (0.7,0);}}
\newcommand{\klarrowtail}{\klarrow[>->]}
\newcommand{\kltwoheadarrow}{\klarrow[->>]}
\newcommand{\Surj}{\mathrm{Surj}} \newcommand{\Inj}{\mathrm{Inj}}
\newcommand{\Iso}{\mathrm{Iso}} \newcommand{\Eps}{\mathrm{Eps}}
 \newcommand{\Tot}{\mathrm{Tot}}
\newcommand{\Inv}{\mathrm{Inv}}
\newcommand{\word}[1]{\triangleright #1 \triangleleft}

\DeclareMathOperator{\len}{length}
\DeclareMathOperator{\colen}{oplength}
\newcommand{\factorization}[2]{\left. {#1} \middle/ {#2} \right.}
\newcommand{\A}{\mathcal{A}} \newcommand{\B}{\mathcal{B}}
\newcommand{\E}{\mathcal{E}} \newcommand{\I}{\mathcal{I}}
\renewcommand{\H}{\mathcal{H}} 
 \renewcommand{\L}{\mathcal{L}}
\newcommand{\M}{\mathcal{M}} \newcommand{\N}{\mathbb{N}}
\renewcommand{\O}{\mathcal{O}} 
\newcommand{\R}{\mathbb{R}} \newcommand{\T}{\mathcal{T}}
\newcommand{\In}{\mathtt{in}} \newcommand{\Out}{\mathtt{out}}
\newcommand{\St}{\mathtt{st}} \newcommand{\card}[1]{\left| #1 \right|}
\newcommand{\Eval}[1]{\textsc{Eval}_{#1}}
\newcommand{\Equiv}[1]{\textsc{Equiv}_{#1}}

\begin{document}

\title[Active learning of monoidal transducers]{Active learning of deterministic
  transducers\texorpdfstring{\\}{} with outputs in arbitrary monoids}

\titlecomment{An extended abstract of this work was presented at CSL
  2024~\cite{aristoteActiveLearningDeterministic2024}. An extended version also
  discussing minimization is available on
  HAL~\cite{aristoteFunctorialApproachMinimizing2023}.}

\author[Q.~Aristote]{Quentin Aristote\lmcsorcid{0009-0001-4061-7553}}
\address{Université Paris Cité, CNRS, Inria, IRIF, F-75013, Paris, France}
\email{quentin.aristote@irif.fr}

\begin{abstract}
  We study monoidal transducers, transition systems arising as deterministic
  automata whose transitions also produce outputs in an arbitrary monoid, for
  instance allowing outputs to commute or to cancel out. We use the categorical
  framework for minimization and learning of Colcombet, Petri\c{s}an and Stabile
  to recover the notion of minimal transducer recognizing a language, and give
  necessary and sufficient conditions on the output monoid for this minimal
  transducer to exist and be unique (up to isomorphism). The categorical
  framework then provides an abstract algorithm for learning it using membership
  and equivalence queries, and we discuss practical aspects of this algorithm's
  implementation.
\end{abstract}

\maketitle

\section{Introduction}
\label{sec:introduction}

\emph{Transducers} are (possibly infinite) transition systems that take input
words over an input alphabet and translate them to some output words over an
output alphabet. They are numerous ways to implement them, but here we focus on
\emph{subsequential transducers}, i.e deterministic automata whose transitions
also produce an output (see \Cref{fig:transducers} for an example). They are
used in diverse fields such as compilers \cite{fischerCraftingCompiler2010},
linguistics \cite{kaplanRegularModelsPhonological1994}, or natural language
processing \cite{knightApplicationsWeightedAutomata2009}.

Two subsequential transducers are considered equivalent when they
\emph{recognize} the same \emph{subsequential function}, that is if, given the
same input, they always produce the same output. A natural question is thus
whether there is a (unique) minimal transducer recognizing a given function (a
transducer with a minimal number of states and which produces its ouput as early
as possible), and whether this minimal transducer is computable. The answer to
both these questions is positive when there exists a finite subsequential
transducer recognizing this function: the minimal transducer can then for
example be computed through minimization
\cite{choffrutMinimizingSubsequentialTransducers2003}.

\paragraph*{Active learning of transducers.}
\label{par:vilar_algorithm}

Another method for computing a minimal transducer is to learn it through Vilar's
algorithm \cite{vilarQueryLearningSubsequential1996}, a generalization to
transducers of Angluin's L*-algorithm, which learns the minimal deterministic
automaton recognizing a language \cite{angluinLearningRegularSets1987}. Vilar's
algorithm thus relies on the existence of an oracle which may answer two types
of queries, namely:
\begin{itemize}
\item \emph{membership queries}: when queried with an input word, the oracle
  answers with the corresponding expected output word;
\item \emph{equivalence queries}: when queried with a \emph{hypothesis
    transducer}, the oracle answers whether this transducer recognizes the
  target function, and, if not, provides a counter-example input word for which
  this transducer is wrong.
\end{itemize}
The basic idea of the algorithm is to use the membership queries to infer
partial knowledge of the target function on a finite subset of input words, and,
when some \emph{closure} and \emph{consistency} conditions are fulfilled, use
this partial knowledge to build a hypothesis transducer to submit to the oracle
through an equivalence query: the oracle then either confirms this transducer is
the right one, or provides a counter-example input word on which more knowledge
of the target function should be inferred.

\begin{figure}[h]
  \centering
  \begin{subfigure}{0.3\linewidth}
    \centering
    \includefigure{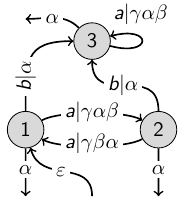}
    \subcaption{A monoidal transducer $\A$}
    \label{fig:learned_transducer}
  \end{subfigure}
  \hspace*{30pt}
  \begin{subfigure}{0.5\linewidth}
    \centering
    \includefigure{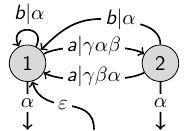}
    \subcaption{A hypothesis transducer built when learning $\A$}
    \label{fig:hypothesis_transducer}
  \end{subfigure}
  \caption{Two transducers: unlike automata, the transitions are also labelled
    with output words.}
  \label{fig:transducers}
\end{figure}

Consider for instance the partial function recognized by the minimal transducer
$\A$ of \Cref{fig:learned_transducer} over the input alphabet $A = \{ a, b \}$
and output alphabet $\Sigma = \{ \alpha, \beta, \gamma \}$. We write this
function $\outval{-}: A^* \rightarrow \Sigma^* \sqcup \{ \bot \}$ (where
$\triangleright$ and $\triangleleft$ are symbols denoting the start and end of a
word), and let $e \in A^*$ and $\varepsilon \in \Sigma^*$ stand for the
respective empty words over these two alphabets. To learn $\A$, the algorithm
maintains a subset $Q \subset A^*$ of prefixes of input words and a subset $T
\subset A^*$ of suffixes of input words, and keeps track of the restriction of
$\outval{-}$ to words in $QT \cup QAT$. The prefixes in $Q$ will be made into
states of the hypothesis transducer, and two prefixes $q, q' \in Q$ will
correspond to two different states if there is a suffix $t \in T$ such that
$\outval{qt} \neq \outval{q't}$. Informally, closure then holds when for every
state $q \in Q$ and input letter $a \in A$ an $a$-transition to some state $q'
\in Q$ can always be built; consistency holds when there is always at most one
consistent choice for such a $q'$ and when the newly-built $a$-transition can be
equipped with an output word. The execution of the learning algorithm for the
function recognized by $\A$ would thus look like the following.

The algorithm starts with $Q = T = \{ e \}$ only consisting of the empty input
word. In a hypothesis transducer, we would want $e \in Q$ to correspond to the
initial state, and the output value produced by the initial transition to be the
longest common prefix $\Lambda(e)$ of each $\outval{et}$ for $t \in T$, here
$\Lambda(e) = \alpha$. But the longest common prefix $\Lambda(a)$ of each
$\outval{at}$ for $t \in T$ is $\gamma \alpha \beta \alpha$, of
which $\Lambda(e)$ is not a prefix: it is not possible to make the output of the
first $a$-transition so that following the initial transition and then the
$a$-transition produces a prefix of $\Lambda(a)$! This is a first kind of
consistency issue, which we solve by adding $a$ to $T$, turning $\Lambda(e)$
into the empty output word $\varepsilon$ and $\Lambda(a)$ into $\gamma \alpha
\beta$.

Now $Q = \{ e \}$ and $T = \{ e, a \}$. The initial transition should go into
the state corresponding to $e$ and output $\Lambda(e) = \varepsilon$, the final
transition from this state should output $\Lambda(e)^{-1}\outval{e} = \alpha$,
the $a$-transition from this state should output $\Lambda(e)^{-1}\Lambda(a) =
\gamma \alpha \beta$, and this $a$-transition followed by a final transition
should output $\Lambda(e)^{-1}\outval{a} = \gamma \alpha \beta \alpha$. This
$a$-transition should moreover lead to a state from which another $a$-transition
followed by a final transition outputs $\Lambda(a)^{-1}\outval{aa} = \gamma
\beta \alpha^2$: in particular, it cannot lead back to the state corresponding
to $e$, because $\gamma \beta \alpha^2 \neq \gamma \alpha \beta \alpha$. But
this state is the only state accounted for by $Q$, so now we have no candidate
for its successor when following the $a$-transition! This is a closure issue,
which we solve by adding $a$ to $Q$, the corresponding new state then being the
candidate successor we were looking for.

Once $Q = T = \{ e, a \}$, there are no closure nor consistency issues and we
may thus build the hypothesis transducer given by
\Cref{fig:hypothesis_transducer}: it coincides with $\A$ on $QA \cup QAT$.
Submitting it to the oracle we learn that this transducer is not the one we are
looking for, and we get as counter-example the input word $bb$, which indeed
satisfies $\L(\triangleright bb \triangleleft) = \bot$ and yet for which our
hypothesis transducer produced the output word $\alpha^3$: we thus add $bb$ and
its prefixes to $Q$.

With $Q = \{ e, a, b, bb \}$ and $T = \{ e, a \}$ there is another kind of
consistency issue, because the states corresponding to $e$ and $b$ are not
distinguished by $T$ ($\Lambda(e)^{-1}\outval{et} = \Lambda(b)^{-1}\outval{bt}$
for all $t \in T$) and should thus be merged in the hypothesis transducer, yet
this is not the case of their candidate successors when following an additional
$b$-transition ($\outval{ebe} = \alpha$ yet $\outval{bbe} = \bot$ is undefined)!
This issue is solved by adding $b$ to $T$, after which there are again no
closure nor consistency issues and we may thus build $\A$ as our new hypothesis
transducer. The algorithm finally stops as the oracle confirms that we found the
right transducer.

\paragraph*{Transducers with outputs in arbitrary monoids.}

In the example above we assumed the output of the transducer consisted of words
over the output alphabet $\Sigma = \{ \alpha, \beta, \gamma \}$, that is of
elements of the free monoid $\Sigma^*$. But in some contexts it may be relevant
to assume that certain output words can be swapped or can cancel each other out.
In other words, transducers may be considered to be \emph{monoidal} and have
output not in a free monoid, but in a quotient of a free monoid. An example of a
non-trivial family of monoids that should be interesting to use as the output of
a transducer is the family of trace monoids, that are used in concurrency theory
to model sequences of executions where some jobs are independent of one another
and may thus be run asynchronously: transducers with outputs in trace monoids
could be used to programatically schedule jobs. Algebraically, trace monoids are
just free monoids where some pairs of letters are allowed to commute. For
instance, the transducers of \Cref{fig:transducers} could be considered under
the assumption that $\alpha\beta = \beta\alpha$, in which case the states $1$
and $2$ would have the same behavior.

This raises the question of the existence and computability of a minimal
monoidal transducer recognizing a function with output in an arbitrary monoid.
In \cite{gerdjikovGeneralClassMonoids2018}, Gerdjikov gave some conditions on
the output monoid for minimal monoidal transducers to exist and be unique up to
isomorphism, along with a minimization algorithm that generalizes the one for
(non-monoidal) transducers. This question had also been addressed in
\cite{eisnerSimplerMoreGeneral2003}, although in a less satisfying way as the
minimization algorithm relied on the existence of stronger oracles. Yet, to the
best of the author's knowledge, no work has addressed the problem of learning
minimal monoidal transducers through membership and equivalence queries.

As all monoids are quotients of free monoids, a first solution would of course
be to consider the target function to have output in a free monoid, learn the
minimal (non-monoidal) transducer recognizing this function using Vilar's
algorithm, and only then consider the resulting transducer to have output in a
non-free monoid and minimize it using Gerdjikov's minimization algorithm. But
this solution is unsatisfactory as, during the learning phase, it may introduce
states that will be optimized away during the minimization phase. For instance,
learning the function recognized by the transducer $\A$ of
\Cref{fig:learned_transducer} with the assumption that $\alpha\beta =
\beta\alpha$ would first produce $\A$ itself before having its states $1$ and
$2$ merged during the minimization phase. Worse still, it is possible to find a
partial function with output in a finitely presented monoid
and recognized by a finite monoidal transducer, but for which Vilar's algorithm
may not terminate if the oracle does not choose its respones carefully:
\begin{lem}
  \label{lemma:vilar_infinite_run}
  Let $A = \{ a \}$, $\Sigma = \{ \alpha, \beta, \gamma \}$, let
  $\Sigma^*/{\sim}$ be the monoid given by the presentation $\langle
  \Sigma \mid \alpha \beta = \beta \alpha \rangle$ and let $\pi: \Sigma^*
  \rightarrow \Sigma^*/{\sim}$ be the corresponding quotient. Consider the
  function $f: A^* \rightarrow \Sigma^*/{\sim}$ that maps $a^n$ to $\alpha^n
  \beta^n \gamma = (\alpha \beta)^n \gamma$.

  $f$ is recognized by a finite transducer with outputs in $\Sigma^*/{\sim}$,
  yet learning a transducer that recognizes any function $f': A^* \rightarrow
  \Sigma^*$ such that $f = f' \circ \pi$ with Vilar's algorithm will never
  terminate if the oracle replies to the membership query for $a^n$ with
  $\alpha^n \beta^n \gamma \in \Sigma^*$ (which differs from $(\alpha \beta)^n
  \gamma$ in $\Sigma^*$ but not in $\Sigma^*/{\sim}$).
\end{lem}
\begin{proof}
  $f$ is recognized by the $(\Sigma^*/{\sim})$-transducer
  (\Cref{def:monoidal_transducer}) with one state $s$ that is initial, initial
  value $\upsilon_0 = \varepsilon$, transition function $a \odot s = (\alpha
  \beta, s)$ and termination function $t(s) = \gamma$.

  Consider now the run of Vilar's algorithm
  (\Cref{alg:learning_monoidal_transducers} with $M = \Sigma^*$ and
  $\invertibles{M} = \{ e \}$) with the oracle answering the membership query
  for $a^n$ with $f'(a^n) = \alpha^n \beta^n \gamma$. We start with $Q = T =
  \{ e \}$, $\Lambda(e) = \gamma$, $\Lambda(a) = \alpha \beta \gamma$ and
  $R(e,e) = R(a,e) = \varepsilon$, where $\Lambda(w)$ for $w \in Q \cup QA$ is
  the longest common prefix of the $\{ f'(wt) \mid t \in T \}$ and $R(w,t)$ is
  the suffix such that $f'(wt) = \Lambda(w)R(w,t)$.

  Since $\Lambda(e) = \gamma$ is not a prefix of $\Lambda(a) = \alpha \beta
  \gamma$, there is a consistency issue and we add $a$ to $T$. Taking $n = 0$,
  we are now in the configuration $C_n$ given by $Q = Q_n = \{ a^k \mid k \le
  n \}$, $T = \{ e, a \}$, and $\Lambda(a^k) = \alpha^k$, $R(a^k, e) = \beta^k
  \gamma$ and $R(a^k,a) = \alpha \beta^{k+1} \gamma$ for every $k \le n + 1$
  (since the oracle replies to the membership query for $a^k$ with $\alpha^k
  \beta^k \gamma$ and to that for $a^ka$ with
  $\alpha^{k+1}\beta^{k+1}\gamma$).

  Suppose now we are in the configuration $C_n$ for some $n \in \N$. Then
  there is a closure issue, since for all $k \le n$, $R(a^k,e) = \beta^k\gamma
  \neq \beta^{n+1} \gamma = R(a^{n+1},e)$. We thus add $a^{n+1}$ to $Q$. To
  compute $\Lambda(a^{n+2})$, $R(a^{n+2},e)$ and $R(a^{n+2},a)$, we make a
  membership query for $a^{n+3}$: the oracle answers with $f'(a^{n+3}) =
  \alpha^{n+3} \beta^{n+3} \gamma$. The longest common prefix of $f'(a^{n+2})
  = \Lambda(a^{n+1})R(a^{n+1},a) = \alpha^{n+2}\beta^{n+2}\gamma$ and
  $f'(a^{n+2}a) = \alpha^{n+3} \beta^{n+3} \gamma$ is thus $\alpha^{n+2}$, and
  the corresponding suffixes are $R(a^{n+2},e) = \beta^{n+2} \gamma$ and
  $R(a^{n+2},a) = \alpha \beta^{n+3} \gamma$: we are now in the configuration
  $C_{n+1}$.

  Hence the run of the algorithm never terminates and never even reaches an
  equivalence query, as it must first go through all the configurations $C_n$
  for $n \in \N$.
\end{proof}
The idea of finding such an example was suggested by an anonymous reviewer whom
the author thanks.

A more satisfactory solution would hence do away with the minimization phase and
instead use the assumptions on the output monoid during the learning phase to
directly produce the minimal monoidal transducer.

\paragraph*{Structure and contributions.}

In this work we thus study the problem of generalizing Vilar's algorithm to
monoidal transducers. To this aim, we first recall in
\Cref{sec:categorical_framework} the categorical framework of Colcombet,
Petri\c{s}an and Stabile for learning minimal transition systems
\cite{colcombetLearningAutomataTransducers2020}. This framework encompasses both
Angluin's and Vilar's algorithms, as well as a similar algorithm for weighted
automata
\cite{bergadanoLearningBehaviorsAutomata1994,bergadanoLearningBehaviorsAutomata1996}.
We use this specific framework because, while others exist, they either do not
encompass transducers or require stronger assumptions
\cite{barloccoCoalgebraLearningDuality2019,urbatAutomataLearningAlgebraic2020,vanheerdtLearningAutomataSideEffects2020}.
In \Cref{sec:category_monoidal_transducers} we then instantiate this framework
to retrieve monoidal transducers as transition systems whose state-spaces live
in a certain category (\Cref{sec:category_monoidal_transducers:as_functors}).
Studying the existence of specific structures in this category --- namely,
powers of the terminal object
(\Cref{sec:category_monoidal_transducers:initial_final_objects}) and
factorization systems
(\Cref{sec:category_monoidal_transducers:factorization_systems}) --- we then
give conditions for the framework to apply and hence for the minimal monoidal
transducers to exist and be computable.

This paper's contributions are thus the following:
\begin{itemize}
\item necessary and sufficient conditions on the output monoid for the
  categorical framework of Colcombet, Petri\c{s}an and Stabile to apply to
  monoidal transducers are given;
\item these conditions mostly overlap those of Gerdjikov, but are nonetheless
  not equivalent: in particular, they extend the class of output monoids for
  which minimization is known to be possible, although with a possibly worse
  complexity bound;
\item practical details on the implementation of the abstract monoidal
  transducer-learning algorithm that results from the categorical framework are
  given;
\item in particular, additional structure on the category in which the framework
  is provides a neat categorical explanation to both the different
  kinds of consistency issues that arise in the learning algorithm and the main
  steps that are taken in every transducer minimization algorithm.
\end{itemize}

\section{Categorical approach to learning minimal automata}
\label{sec:categorical_framework}

In this section we recall (and extend) the definitions and results of Colcombet,
Petri\c{s}an and Stabile \cite{colcombetAutomataMinimizationFunctorial2020,
  colcombetLearningAutomataTransducers2021}. We assume basic knowledge of
category theory \cite{maclaneCategoriesWorkingMathematician1978}, but we also
focus on the example of deterministic complete automata and on the
counter-example of non-deterministic automata. We do not explain it here but the
framework also applies to weighted automata
\cite{bergadanoLearningBehaviorsAutomata1994,bergadanoLearningBehaviorsAutomata1996,schutzenbergerDefinitionFamilyAutomata1961}
and (non-monoidal) transducers (as generalized in
\Cref{sec:category_monoidal_transducers}).

\subsection{Automata and languages as functors}
\label{sec:categorical_framework:automata_languages}

Consider the graph $\smash{\In \xrightarrow[]{\triangleright}
  \overset{{\text{\normalsize$\overset{a}\curvearrowleft$}}}{\St}
  \xrightarrow[]{\triangleleft} \Out}$ where $a$ ranges in the \emph{input
  alphabet} $A$, and let $\I$, the \emph{input category}, be the category that
it generates: the objects are the vertices of this graph and the morphisms paths
between two vertices. $\I$ represents the basic structure of automata as
transition systems: $\St$ represents the state-space, $\triangleright$ the
initial configuration, each $a: \St \rightarrow \St$ the transition along the
corresponding letter, and $\triangleleft$ the output values associated to each
state. An automaton is then an instantiation of $\I$ in some output category:

\begin{defi}[$(\C,X,Y)$-automaton]
  \label{def:C-automaton}
  Given an output category $\C$, a \emph{$(\C, X, Y)$-automaton} is a functor
  $\A: \I \rightarrow \C$ such that $\A(\In) = X$ and $\A(\Out) = Y$.
\end{defi}

\begin{exa}[(non-)deterministic automata]
  If $1 = \{ * \}$ and $2 = \{ \bot, \top \}$, a $(\Set, 1, 2)$-automaton $\A$
  is a (possibly infinite) deterministic complete automaton: it is given by a
  state-set $S = \A(\St)$, transition functions $\A(a): S \rightarrow S$ for
  each $a \in A$, an initial state $s_0 = \A(\triangleright)(*) \in S$ and a set
  of accepting states $F = \{ s \in S \mid \A(\triangleleft)(s) = \top \}
  \subseteq S$. Similarly, a $(\Rel, 1, 1)$-automaton (where $\Rel$ is the
  category of sets and relations between them) is a (possibly infinite)
  non-deterministic automaton: it is given by a state-set $S = \A(\St)$,
  transition relations $\A(a) \subseteq S \times S$, a set of initial states
  $\A(\triangleright) \subseteq 1 \times S \cong S$ and a set of accepting
  states $\A(\triangleleft) \subseteq S \times 1 \cong S$.

  % Finally, if $\K$ is a field we may see $\K$-weighted automata as functors $\A$
  % from $\I$ to the category of $\K$-vector-spaces, $\Vec{\K}$, such that
  % $\A(\In) = \A(\Out) = \K$.

  % , and if $B$ is an output alphabet, we may see deterministic
  % transducers as functors from $\I$ to the Kleisli category of free algebras for
  % the monad $\T_{B^*} X = B^* \times X + 1$, $\Kl(\T_{B^*})$, such that $\A(\In)
  % = \A(\Out) = 1$. This last example will be detailed and generalized in
  % \Cref{def:monoidal_transducers}.
\end{exa}

Let $\O$, the \emph{category of observable inputs}, be the full subcategory of
$\I$ on $\In$ and $\Out$.

\begin{defi}[$(\C,X,Y)$-language]
  \label{def:C-language}
  In the same way, we define a \emph{$(\C,X,Y)$-language} to be a functor $\L:
  \O \rightarrow \C$ such that $\L(\In) = X$ and $\L(\Out) = Y$. In particular,
  composing an automaton $\A: \I \rightarrow \C$ with the embedding $\iota: \O
  \hookrightarrow \I$, we get the language $\L_{\A} = \A \circ \iota$
  \emph{recognized} by $\A$.
\end{defi}

In other words, a language is the data of two objects $X = \L(\In)$ and $Y =
\L(\Out)$ in $\C$ and, for each word $w \in A^*$, of a morphism $\L(\word{w}): X
\rightarrow Y$.

\begin{exa}[languages] The language recognized by a
  $(\Set, 1, 2)$-automaton $\A$ is the language recognized by the corresponding
  complete deterministic automaton: for a given $w \in A^*$,
  $\L_{\A}(\word{w})(*) = \top$ if and only if $w$ belongs to the language, and
  the equality $\L_{\A}(\word{w}) = \A(\word{w}) =
  \A(\triangleright)\A(w)\A(\triangleleft)$ means that we can decide whether $w$
  is in the language by checking whether the state we get in by following $w$
  from the initial state is accepting. Such a language could also be seen as a
  set of relations $\L(\word{w}) \subseteq 1 \times 1$, where $*$ is related to
  itself if and only if $w$ belongs to $\L$. The language recognized by a
  $(\Rel, 1, 1)$-automaton is thus the language recognized by the corresponding
  non-deterministic automaton.

  % Similarly, when $\C$ is $\Vec{\K}$
  % % or $\Kl(B^* \times - + 1)$
  % the corresponding notion of language is that recognized by a weighted
  % automaton.
\end{exa}

\begin{defi}[category of automata recognizing a language]
  \label{def:cat_auto_l}
  Given a category $\C$ and a language $\L: \O \rightarrow \C$, we define
  the category $\Auto{\L}$ whose objects are $(\C, \L(\In), \L(\Out))$-automata
  $\A$ recognizing $\L$, and whose morphisms $\A \rightarrow \A'$ are natural
  transformations whose components on $\L(\In)$ and $\L(\Out)$ are the identity.
  In other words, a morphism of automata is given by a morphism $f: \A(\St)
  \rightarrow \A'(\St)$ in $\C$ such that $\A'(\triangleright) = f \circ
  \A(\triangleright)$ (it preserves the initial configuration), $\A'(a) \circ f
  = f \circ \A(a)$ (it commutes with the transitions), and $\A'(\triangleleft)
  \circ f = \A(\triangleleft)$ (it preserves the output values).
\end{defi}

\subsection{Factorization systems and the minimal automaton recognizing a
  language}
\label{sec:categorical_framework:minimal_automaton}

\begin{defi}[factorization system]
  \label{def:factorization_system}
  In a category $\C$, a \emph{factorization system} $(\E, \M)$ is the data of a
  class of \emph{$\E$-morphisms} (represented with $\twoheadrightarrow$) and a
  class of \emph{$\M$-morphisms} (represented with $\rightarrowtail$) $\M$ such
  that

  \begin{itemize}
  \item every arrow $f$ in $\C$ may be \emph{factored} as $f = m \circ e$ with
    $m \in \M$ and $e \in \E$;
  \item $\E$ and $\M$ are both stable under composition;
  \item for every commuting diagram as below where $m \in \M$ and $e \in
    \E$ there is a unique diagonal fill-in $d: Y_1 \rightarrow Y_2$ such that
    $u = d \circ e$ and $v = m \circ d$.
  \end{itemize}

  \begin{center}
    % https://tikzcd.yichuanshen.de/#N4Igdg9gJgpgziAXAbVABwnAlgFyxMJZABgBpiBdUkANwEMAbAVxiRAA0QBfU9TXfIRQBGclVqMWbAJoB9Yd14gM2PASJlh4+s1aIQcgEyK+qwUVFbqOqfoBa3cTCgBzeEVAAzAE4QAtkhkIDgQSKIgABYwdFBsOADuEFExCNaSeiCsPF6+AYiG1CFIAMyFdFgMceWVabpsAdQMdABGMAwACvxqQiDeWC4ROCYgPv6BhaH5tbYgTCCNLW2dZur6fQND2SO5YRMl0xk0jlxAA
    \begin{tikzcd}
      X \arrow[r, "e", two heads] \arrow[d, "u"'] & Y_1 \arrow[d, "v"]
      \arrow[dl, dashed, "d"] \\
      Y_2 \arrow[r, "m"', tail] & Z
    \end{tikzcd}
  \end{center}
\end{defi}

\begin{exa}
  \label{example:factorization_systems}
  In $\mathbf{Set}$ surjective and injective functions form a factorization
  system $(\Surj,\Inj)$ such that a map $f: X \rightarrow Y$ factors through
  its image $f(X) \subseteq Y$. In $\Rel$ a factorization system is given by
  $\E$-morphisms those relations $r: X \rightarrow Y$ such that every $y \in Y$
  is related to some $x \in X$ by $r$, and $\M$-morphisms the graphs of
  injective functions (i.e. $m \in \M$ if and only if there is an injective
  function $f$ such that $(x,y) \in m \iff y = f(x)$). A relation $r: X
  \rightarrow Y$ then factors through the subset $\{ y \in Y \mid \exists x \in
  X, (x,y) \in R \} \subseteq Y$.

  % We write this factorization system $(\Surj, \Det \cap \Inj)$ (for surjections
  % and deterministic injections).
\end{exa}

\begin{lem}[factorization system on $\Auto{\L}$
  {\cite[Lemma~2.8]{colcombetAutomataMinimizationFunctorial2020}}]
  \label{lemma:factorization_system_auto}
  Given $\L: \O \rightarrow \C$ and $(\E,\M)$ a factorization system on $\C$,
  $\Auto{\L}$ has a factorization given by those natural transformations whose
  components are respectively $\E$- and $\M$-morphisms.
\end{lem}

Because of this last result, we use $(\E,\M)$ to refer both to a factorization
on $\C$ and to its extensions to categories of automata.

\begin{defi}[minimal object]
  \label{def:minimal_object}
  When a category $\C$, equipped with a factorization system $(\E, \M)$, has
  both an initial object $I$ and a final object $F$ (for every object $X$ there
  is exactly one morphism $I \rightarrow X$ and one morphism $X \rightarrow F$),
  we define its \emph{$(\E,\M)$-minimal object} $\Min{}$ to be the one that
  $(\E, \M)$-factors the unique arrow $I \rightarrow F$ as $I \twoheadrightarrow
  \Min{} \rightarrowtail F$. For every object $X$ we also define $\Reach X$ and
  $\Obs X$ by the $(\E, \M)$-factorizations $I \twoheadrightarrow \Reach X
  \rightarrowtail X$ and $X \twoheadrightarrow \Obs X \rightarrowtail F$.
\end{defi}

\begin{prop}[uniqueness of the minimal object {\cite[Lemma~2.3]{colcombetAutomataMinimizationFunctorial2020}}]
  \label{proposition:minimal_object_infimum_divisibility}
  The minimal object of a category $\C$ is unique up to isomorphism, so that for
  every other object $X$, $\Min \cong \Obs(\Reach X) \cong \Reach(\Obs X)$:
  there are in particular spans $X \leftarrowtail \Reach X \twoheadrightarrow
  \Min$ and co-spans $\Min \rightarrowtail \Obs X \twoheadleftarrow X$. It is in
  that last sense that $\Min$ is $(\E,\M)$-smaller than every other object $X$,
  and is thus minimal.
  % we have commuting diagrams
  % \begin{center}
  %   % https://tikzcd.yichuanshen.de/#N4Igdg9gJgpgziAXAbVABwnAlgFyxMJZABgBoBGAXVJADcBDAGwFcYkQBJEAX1PU1z5CKchWp0mrdgB1pAJRj0AxgAsABAA0efEBmx4CRAEyli4hizaIQW3v31DjpI+clWQsgLJYwwbtvtBQxQAZjEaCylrWQB5ACM4TQDdAQNhZAAWcIlLdgAxHnEYKABzeCJQADMAJwgAWyQyEBwIJFEQFUUodhwAdwhO+igEOxAa+raaFqQTDq72SDA2KfosRh7V9dHxhsR26cQwuaGe-sHh5J2ZqdbELOPu6z6BrpGdK7ubpABWFbWN-7bWq7I4He44TYAraUbhAA=
  %   \begin{tikzcd}
  %     &                                                          & X \arrow[rd, two heads] &                        &   \\
  %     I \arrow[r, two heads] & \Reach X \arrow[ru, tail] \arrow[rd, two heads] &                         & \Obs X \arrow[r, tail] & F \\
  %     & & \Min{} \arrow[ru, tail] & &
  %   \end{tikzcd}
  % \end{center}
\end{prop}

\begin{exa}[initial, final and minimal automata
  {\cite[Example~3.1]{colcombetAutomataMinimizationFunctorial2020}}]
  \label{example:minimal_automata}
  Since $\mathbf{Set}$ is complete and cocomplete, the category of
  $(\mathbf{Set}, 1, 2)$-automata recognizing a language $\L: \I \rightarrow
  \mathbf{Set}$ has an initial, a final and a minimal object. The initial
  automaton has state-set $A^*$, initial state $\varepsilon \in A^*$, transition
  functions $\delta_a(w) = wa$ and accepting states the $w \in A^*$ such that
  $w$ is in $\L$. Similarly, the final automaton has state-set $2^{A^*}$,
  initial state $\L$, transition functions $\delta_a(L) = a^{-1}L$ and accepting
  states the $L \in 2^{A^*}$ such that $\varepsilon$ is in $L$. The minimal
  automaton for the factorization system of \Cref{example:factorization_systems}
  thus has the Myhill-Nerode equivalence classes for its states. It is unique up
  to isomorphism and its $(\E, \M)$-minimality ensures that it is the complete
  deterministic automaton with the smallest state-set that recognizes $\L$: it
  is in particular finite as soon as $\L$ is recognized by a finite automaton.

  % For the same reason, the category of $(\mathbf{Vec}, \K, \K)$-automata
  % recognizing a language $\L: \I \rightarrow \mathbf{Vec}$ also has a minimal
  % object that corresponds to the minimal $\K$-weighted automaton recognizing
  % $\L$ \cite{schutzenbergerDefinitionFamilyAutomata1961}.

  % The example of transducers seen as automata over the Kleisli category for
  % the monad $X \mapsto B^* \times X + 1$ is more involved: this category has
  % all coproducts, but not necessarily products. Yet, it happens that $1$ does
  % have countable powers, and hence minimization is still possible for
  % $(\Kl(\T), 1, 1)$-automata: with the right factorization system, we
  % retrieve Choffrut's minimal transducer
  % \cite{choffrutMinimizingSubsequentialTransducers2003}.

  On the contrary, there is no good notion of a unique minimal non-deterministic
  automaton recognizing a regular ($(\Rel,1,1)$-) language $\L$. $\Auto{\L}$
  does have an initial and a final object: the initial automaton is the initial
  deterministic automaton recognizing $\L$, and the final automaton is the
  (non-deterministic) transpose of this initial automaton. But there is no
  factorization system that gives rise to a meaningful minimal object: the
  latter is obtained as the factorization of the relation $\{ (u,v) \in \Sigma^*
  \times \Sigma^* \mid uv \in \L \}$, and, for the factorization system of
  \Cref{example:factorization_systems} for instance, its state-set would then be
  the set of suffixes of words in $\L$.
\end{exa}

Notice how in \Cref{example:minimal_automata} the initial and final
$(\Set,1,2)$-automata have for respective state-sets $A^*$, the disjoint union
of $|A^*|$ copies of $1$, and $2^{A^*}$, the cartesian product of $|A^*|$ copies
of $2$. A similar result holds for non-deterministic automata and generalizes as
\Cref{thm:minimal_automaton}, itself summarized by the diagram below,
where $\kappa$ and $\pi$ are the canonical inclusion and projections and $[-]$
and $\langle - \rangle$ are the copairing and pairing of arrows.

\begin{center}
  % https://tikzcd.yichuanshen.de/#N4Igdg9gJgpgziAXAbVABwnAlgFyxMJZABgBoBGAXVJADcBDAGwFcYkQAdDgGQAouAkmACUIAL6l0mXPkIpypYtTpNW7LgGMIaAE7QA+sACCAPQBUYgARc+gkeMkgM2PASIKATMoYs2iThy6BsbmVjb8HADyzDiiElIuskQeFN6qfgG2UTFxjs4ybvKpND5q-lwAslhgEdzCEQDKseLKMFAA5vBEoABmegC2SGQgOBBICiAARjBgUEgAzMOlGVwA1vRoaPSGXDBo2IwEYg69A+M0ows007MLS+nsyOFcODpY9GDtjDBv7QAWOEsAHdrBxXu9Pt9vj1YpRDCCuNVLKYLCcQH0IINEBNLtiaIdtO4ABxkHpMOAwZSMejTRgABWkrjkIF+AJAJQe5Q4WiCUHhoKRKLCHHWm22wCB9GO8XRZ0Qw1xKSmMzmiAAtIsOb51BxqZCYKCsuCPl8flh-oCEWC3iaoTAYcJQTpbTB+YiwMjQuyQNTaQzEoUWea2TKMVilbj5tcVUgNfdtVy0Fgdhw9gcjt7fTB6Yykv5WTg0WGkBGxoglQS0ERSeTKfiadn-QVmQXvcsdby3RxBaFQUnDPQgdLHMW8SMywAWGh-GD0VUjIEQGdzhChuVT8cli70LCMdg4Hd7teYoYXSfHrEbyNiShiIA
  \tikzset{
    rotatebelow/.style={anchor=north, rotate=-30, inner sep=.7em},
    rotateabove/.style={anchor=south, rotate=-30, inner sep=.7em}
  }
  \begin{tikzcd}[column sep=huge]
    & \coprod_{A^*} \L(\In) \arrow[rd, "{[\outval{ w }]_{w \in A^*}}" rotateabove, bend left] \arrow["\coprod_{w \in A^*} \kappa_{wa}"', loop, distance=1.5em, in=120, out=60] \arrow[d, two heads] &          \\
    \L(\In) \arrow[ru, "\kappa_{\varepsilon}", bend left] \arrow[rd, "\langle
    \outval{ w } \rangle_{w \in A^*}"' rotatebelow, bend right] \arrow[r] & \Min(\L)(\St) \arrow[d, tail] \arrow[r]                                                                                                                                                               & \L(\Out) \\
    & \prod_{A^*} \L(\Out) \arrow[ru, "\pi_{\varepsilon}"', bend right]
    \arrow["\prod_{w \in A^*} \pi_{aw}"', loop, distance=1.5em, in=300, out=240]
    &
  \end{tikzcd}
\end{center}

\begin{thmC}[{\cite[Lemma~3.2]{colcombetAutomataMinimizationFunctorial2020}}]
  \label{thm:minimal_automaton}
  Given a countable alphabet $A$ and a language $\L: \I \rightarrow \C$,
  \begin{itemize}
  \item if $\C$ has all countable copowers of $\L(\In)$ then $\Auto{\L}$ has an
    initial object $\A^{init}(\L)$ with $\A^{init}(\L)(\St) = \coprod_{A^*}
    \L(\In)$;
  \item dually if $\C$ has all countable powers of $\L(\Out)$ then $\Auto{\L}$
    has a final object $\A^{final}(\L)$ with $\A^{final}(\L)(\St) = \prod_{A^*}
    \L(\Out)$;
  \item hence when both of the previous items hold and $\C$ comes equipped with
    a factorization system $(\E,\M)$, $\Auto{\L}$ has an $(\E,\M)$-minimal
    object $\Min \L$.
  \end{itemize}
\end{thmC}
\noindent We now have all the ingredients to define algorithms for computing the minimal
automaton recognizing a language. But since we will also want to prove the
termination of these algorithms, we need an additional notion of finiteness.

\begin{defi}[$\E^{op}$- and $\M$-n\oe{}therian objects
  {\cite[Definition~24]{colcombetLearningAutomataTransducers2020}}]
  \label{def:E-artinian_M-noetherian}
  In a category $\C$ equipped with a factorization system $(\E,\M)$, an object
  $X$ is said to be \emph{$\M$-n\oe{}therian} if every strict chain of
  $\M$-subobjects is finite: if $(x_n: X_n \rightarrowtail X)_{n \in \N}$ and
  $(m_n: X_n \rightarrowtail X_{n+1})$ form the commutative diagram
  \begin{center}
    % https://tikzcd.yichuanshen.de/#N4Igdg9gJgpgziAXAbVABwnAlgFyxMJZABgBoBGAXVJADcBDAGwFcYkQANAfWJAF9S6TLnyEU5CtTpNW7buX6CQGbHgJEATJJoMWbRCAA6hgMZQIOBAKGrRm0sSm7ZBjvykwoAc3hFQAMwAnCABbJDIQHAgkAGYaACMYMChYiJx6LEZ2dMyQHRl9EAAPHkUA4LDECUjoxDiQROTU-L12EoUaHKyDLrKQINDwztrq3p6MrJaXEBDSmkZ6RMYABWE1MRBArC8ACxw+gcrR2q1IiezzqcLZjpAFpdXbdQMt3f3rfoqkU6jYzvPxrk+JQ+EA
    \begin{tikzcd}
      &                                                                & X                      \\
      X_0 \arrow[rru, "x_0", tail, bend left] \arrow[r, "m_0"', tail] & X_1 \arrow[ru, "x_1", tail, bend left] \arrow[r, "m_1"', tail] & \cdots \arrow[u, tail]
    \end{tikzcd}
  \end{center}
  then only finitely many of the $m_n$'s may not be isomorphisms. Dually, $X$ is
  \emph{$\E^{op}$-n\oe{}therian} if $X$ is so in $\C^\op$, that is if
  every strict cochain of $\E$-quotients of $X$ is finite.

  An object that is both $\E^{op}$- and $\M$-n\oe{}therian is said to be
  $(\E^{op},\M)$-n\oe{}therian.
\end{defi}

While Colcombet, Petri\c{s}an and Stabile do not give complexity results for
their algorithm, it is straightforward to do so, hence we extend the definition
of $(\E,\M)$-n\oe{}therianity so that it also measures the size of an object in
$\C$.

\begin{defi}[$\E^{op}$- and $\M$-lengths]
  \label{def:co-E_M-length}
  For a fixed $x_0: X_0 \rightarrowtail X$, we call \emph{$\M$-length} of
  $x_0$, written $\len_\M x_0$, the (possibly infinite) supremum of the
  lengths (the number of pairs of consecutive subobjects) of strict chains of
  $\M$-subobjects of $X$ that start with $x_0$. Dually, we call
  \emph{$\E^{op}$-length} of an $\E$-quotient $x_0: X \twoheadrightarrow X_0$ the
  (possibly infinite) quantity $\colen_\E x_0 = \len_{\E^\op} x_0^\op$.
\end{defi}

\begin{exa}
  In $\Set$, $X$ is finite if and only if it is $\Inj$-n\oe{}therian iff it is
  $\Surj^{op}$-n\oe{}therian, and in that case for $Y \subseteq X$ we have $\colen_\Surj(X
  \twoheadrightarrow Y) = \len_\Inj(Y \rightarrowtail X) = |X - Y|$.
  % Similarly, in $\Vec{\K}$ $X$ has finite dimension if and only if it is
  % $\Inj$-n\oe{}therian iff it is $\Surj$-artinian, and in that case for $Y
  % \subseteq X$ we have $\colen_\Surj(X \twoheadrightarrow Y) = \len_\Inj(Y
  % \rightarrowtail X) = \dim X - \dim Y$.

  Note that the $E^{op}$- and $\M$-lengths need not be equal: see for instance the
  factorization system we define for monoidal transducers in
  \Cref{sec:category_monoidal_transducers:factorization_systems}, for which the
  $\E^{op}$- and $\M$-lengths are computed in
  \Cref{lemma:M-noetherianity,lemma:E-artinianity}.
\end{exa}

\subsection{Learning}
\label{sec:categorical_framework:learning}

In this section, we fix a language $\L: \O \rightarrow \C$ and a factorization
system $(\E, \M)$ of $\C$ that extends to $\Auto{\L}$, and we assume that $\C$
has countable copowers of $\L(\In)$ and countable powers of $\L(\Out)$ so that
\Cref{thm:minimal_automaton} applies. Our goal is to compute $\Min\L$ with the
help of an oracle answering two types of queries: the function $\Eval{\L}$
processes membership queries, and, for a given input word $w \in A^*$, outputs
$\L(\word{w})$; while $\Equiv{\L}$ processes equivalence queries, and, for a
given hypothesis $(\C, \L(\In), \L(\Out))$-automaton $\A$, decides whether $\A$
recognizes $\L$, and, if not, outputs a counter-example $w \in A^*$ such that
$\L(\word{w}) \neq (\A \circ \iota)(\word{w})$.

For $(\Set, 1, 2)$-automata, if the language is regular this problem is solved
using Angluin's L* algorithm \cite{angluinLearningRegularSets1987}. It works by
maintaining a set of prefixes $Q$ and of suffixes $T$ and, using $\Eval{\L}$,
incrementally building a table $L: Q \times (A \cup \{ \varepsilon \}) \times T
\rightarrow 2$ that represents partial knowledge of $\L$ until it can be made
into a (minimal) automaton. This automaton is then submitted to $\Equiv{\L}$
when some closure and consistency conditions hold: if the automaton is accepted
it must be $\Min\L$, otherwise the counter-example is added to $Q$ and the
algorithm loops over. The \textsc{FunL*} algorithm (\Cref{alg:FunL*},
\cpageref{alg:FunL*}) generalizes this to arbitrary $(\C, \L(\In), \L(\Out))$,
and in particular also encompasses Vilar's algorithm for learning (non-monoidal)
transducers, which was described in \Cref{sec:introduction}
\cite{colcombetLearningAutomataTransducers2020}.

Instead of maintaining a table, the \textsc{FunL*} algorithm maintains a
biautomaton: if $Q \subseteq A^*$ is prefix-closed ($wa \in Q \Rightarrow w \in
Q$) and $T \subseteq A^*$ is suffix-closed ($aw \in T \Rightarrow w \in T$), a
$(Q,T)$-\emph{biautomaton} is, similarly to an automaton, a functor $\A:
\I_{Q,T} \rightarrow \C$, where $\I_{Q,T}$ is now the category freely generated
by the graph
  % https://tikzcd.yichuanshen.de/#N4Igdg9gJgpgziAXAbVABwnAlgFyxMJZABgBpiBdUkANwEMAbAVxiRAB12BJQgX1PSZc+QigCM5KrUYs2nAMo4A+mJD9B2PASIAmSdXrNWiDu0VKdagSAyaRRAMz7pRuewDyTHGqkwoAc3giUAAzACcIAFskMhAcCCQJF1kTThwwrDowfwYYDP8ACxwAAgBHK1CI6MQk+KQ9EDgCrBDvGuoGOgAjGAYABSEtURB8opADGWNTGDRsBm11EHCoxOo6xAamlraAWiTDFJA6CqWq+rWExCdkqZK0jKyc3phWn14gA
  \begin{tikzcd}[cramped, sep=small]
    \In \arrow[r, "\triangleright q"] & \St_1 \arrow[r, "\varepsilon"', shift
    right] \arrow[r, "a", shift left] & \St_2 \arrow[r, "t \triangleleft"] &
    \Out
  \end{tikzcd}
% \end{center}
where $a$, $q$ and $t$ respectively range in $A$, $Q$ and $T$, and where we also
require the diagrams below to commute, the left one whenever $qa \in Q$
and the right one whenever $at \in T$.

\begin{center}
  % https://tikzcd.yichuanshen.de/#N4Igdg9gJgpgziAXAbVABwnAlgFyxMJZABgBpiBdUkANwEMAbAVxiRAB12BJQgX1PSZc+QigCM5KrUYs2nAMo4A+mJD9B2PASISxU+s1aIO7RSrUCQGTSKIAmSdQOzjC5avVWhW0cgDMjtKGcqbKdhYawtooDnpOMkYmZuGe1lG+AXFBLklhEV420cgALKRZzomcAPJMOGpSMFAA5vBEoABmAE4QALZIZCA4EEgS2ZXsOJ1YdGBNDDBTTQAWOAAEAI75Xb391ENIDmMhk9Oz84srqwAU63QAlCDUDHQARjAMAAretsYXdZ7bPqIUb7RAAVniwWMdC23SBh1BEKOrnYMDQ2AY2gBcKQAUGw0QADZITkYU9Xu8voVRCA-rCdog8aDSsiTGiMVjLICkMT8UgAOwkxJrTgnGZzd4wdp1clvT7faK0rDLf5cnGIFmgwWsq50HB3Vaiqbi+bzaX1XhAA
  \begin{tikzcd}
    \In \arrow[r, "\triangleright q"] \arrow[rd, "\triangleright (qa)"'] & \St_1 \arrow[rd, "a"]       & \St_1 \arrow[rd, "a"'] \arrow[r, "\varepsilon"] & \St_2 \arrow[rd, "(at) \triangleleft"] &      \\
    & \St_1 \arrow[r, "\varepsilon"] & \St_2                                        & \St_2 \arrow[r, "t \triangleleft"']    & \Out
  \end{tikzcd}
\end{center}

A $(Q,T)$-biautomaton may thus process a prefix in $Q$ and get in a state in
$\A(\St_1)$, follow a transition along $A \cup \{ \varepsilon \}$ to go in
$\A(\St_2)$, and output a value for each suffix in $T$. The category of
biautomata recognizing $\L_{Q,T}$ ($\L$ restricted to words in $QT \cup QAT$) is
written $\BiAuto{Q,T}{\L}$. A result similar to \Cref{thm:minimal_automaton}
also holds for biautomata
\cite[Lemma~18]{colcombetLearningAutomataTransducers2020}, and the initial and
final biautomata are then made of finite copowers of $\L(\In)$ and finite powers
of $\L(\Out)$ (when these exist). Writing $\factorization{Q}{T}$ for the $(\E,
\M)$-factorization of the canonical morphism $\langle [ \L(\word{qt}) ]_{q \in
  Q} \rangle_{t \in T}: \coprod_Q \L(\In) \rightarrow \prod_T \L(\Out)$, the
minimal biautomaton recognizing $\L_{Q,T}$ then has state-spaces
$(\Min\L_{Q,T})(\St_1) =\factorization{Q}{(T \cup AT)}$ and
$(\Min\L_{Q,T})(\St_2) =\factorization{(Q \cup QA)}{T}$. We also write
$\varepsilon_{Q,T}^{min} = (\Min \L_{Q,T})(\varepsilon)$. The table, represented
by the morphism $\langle [ \L(\word{qt}) ]_{q \in Q} \rangle_{t \in T}$, may be
fully computed using $\Eval{\L}$, and hence so can be the minimal
$(Q,T)$-biautomaton.

A biautomaton $\B$ can then be merged into a hypothesis
$(\C,\L(\In),\L(\Out))$-automaton precisely when $\B(\varepsilon)$ is an
isomorphism, i.e. both an $\E$- and an $\M$-morphism (a factorization system
necessarily satisfies that $\Iso = \E \cap \M$): this encompasses respectively
the closure and consistency conditions that need to hold in the L*-algorithm
(and its variants) for the table that is maintained to be merged into a
hypothesis automaton.

\begin{thmC}[{\cite[Theorem~26]{colcombetLearningAutomataTransducers2020}}]
  \label{thm:FunL*_correction_termination}
  \Cref{alg:FunL*} is correct. If $(\Min \L)(\St)$ is $\M$-n\oe{}therian and
  $\E^{op}$-n\oe{}therian, the algorithm also terminates.
\end{thmC}

While Colcombet, Petri\c{s}an and Stabile do not give a bound on the actual
running time of their algorithm, it is straightforward to extend their proof to
show that the number of updates to $Q$ and $T$ (hence in particular of calls to
$\Equiv{\L}$) is linear in the size of $(\Min \L)(\St)$, itself defined through
\Cref{def:co-E_M-length}. A similar claim is made in \cite[Remark
3.14]{urbatAutomataLearningAlgebraic2020}.

\begin{algorithm*}
  \caption{The \textsc{FunL*}-algorithm}
  \label{alg:FunL*}
  \begin{algorithmic}[1] \REQUIRE $\Eval{\L}$ and $\Equiv{\L}$ \ENSURE
    $\Min(\L)$

    \STATE $Q = T = \{ \varepsilon \}$

    \LOOP

    \WHILE{$\varepsilon_{Q,T}^{min}$ is not an isomorphism}

    \IF{$\varepsilon_{Q,T}^{min}$ is not an $\E$-morphism}

    \STATE find $qa \in QA$ such that $\factorization{Q}{T} \rightarrowtail\factorization{(Q \cup
      \{ qa \})}{T}$ is not an $\E$-morphism; add it\linebreak to $Q$

    \ELSIF{$\varepsilon_{Q,T}^{min}$ is not an $\M$-morphism}

    \STATE find $at \in AT$ such that $\factorization{Q}{(T \cup \{ at \})}
    \twoheadrightarrow\factorization{Q}{T}$ is not an $\M$-morphism; add it\linebreak to $T$

    \ENDIF

    % \STATE update $\Min\L_{Q,T}$ accordingly using $\Eval{\L}$

    \ENDWHILE

    \STATE merge $\Min\L_{Q,T}$ into $\H_{Q,T}\L$

    \IF{$\Equiv{\L}(\H_{Q,T}\L)$ outputs some counter-example $w$}

    \STATE add $w$ and its prefixes to $Q$ \label{alg:FunL*:line:add-counter}

    % \STATE update $\Min\L_{Q,T}$ accordingly using $\Eval{\L}$

    \ELSE

    \RETURN $\H_{Q,T}\L$

    \ENDIF

    \ENDLOOP
  \end{algorithmic}
\end{algorithm*}

\section{The category of monoidal transducers}
\label{sec:category_monoidal_transducers}

We now study a specific family of transition systems, monoidal transducers,
through the lens of category theory, so as to be able to apply the framework of
Colcombet, Petri\c{s}an and Stabile. In
\Cref{sec:category_monoidal_transducers:monoids}, we first rapidly recall the
notion of monoid. We then define the category of monoidal transducers
recognizing a language in \Cref{sec:category_monoidal_transducers:as_functors},
and study how it fits into the framework of \Cref{sec:categorical_framework}:
the initial transducer is given in \Cref{corollary:initial_transducer},
conditions for the final transducer to exist are described in
\Cref{sec:category_monoidal_transducers:initial_final_objects}, and
factorization systems are tackled in
\Cref{sec:category_monoidal_transducers:factorization_systems}.

\subsection{Monoids}
\label{sec:category_monoidal_transducers:monoids}

Let us first recall definitions relating to monoids, and fix some notations.
Most of these are standard in the monoid literature, only coprime-cancellativity
(\Cref{def:cancellativity}) and n\oe{}therianity (\Cref{def:noetherianity}) are
uncommon.

\begin{defi}[monoid] A \emph{monoid} $(M, \varepsilon_M, {\otimes_M})$ is a
  set $M$ equipped with a binary operation ${\otimes_M}$ (often called the
  \emph{product}) that is associative ($\forall u, v, w \in M, u \otimes_M (v
  \otimes_M w) = (u \otimes_M v) \otimes_M w$) and has $\varepsilon_M$ as
  \emph{unit} element ($\forall u \in M, u \otimes_M \varepsilon_M = \varepsilon_M
  \otimes_M u = u$). When non-ambiguous, it is simply written $(M, \varepsilon,
  {\otimes})$ or even $M$, and the symbol for the binary operation may be
  omitted.

  The \emph{dual} of $(M, \varepsilon, {\otimes})$, written $(\dual{M},
  \dual{\varepsilon}, \dual{\otimes})$, has underlying set $\dual{M} = M$ and
  identity $\dual{\varepsilon} = \varepsilon$, but symmetric binary operation:
  $\forall u, v \in M, u \otimes^{op} v = v \otimes u$.
\end{defi}

The dual of a monoid is mainly used here for the sake of conciseness: whenever
we define some ``left-property'', the corresponding ``right-property'' is
defined as the left-property but in the dual monoid. Note that when $M$ is
commutative it is its own dual and the left- and right-properties coincide.

\begin{defi}[invertibility] An element $x$ of a monoid $M$ is
  \emph{right-invertible} when there is a $y \in M$ such that $xy = \varepsilon$,
  and $y$ is then called the \emph{right-inverse} of $x$. It is
  \emph{left-invertible} when it is right-invertible in $\dual{M}$, and the
  corresponding right-inverse is called its \emph{left-inverse}. When $x$ is
  both right- and left-invertible, we say it is \emph{invertible}. In that case
  its right- and left-inverse are equal: this defines its \emph{inverse},
  written $\inv{x}$. The set of invertible elements of $M$ is written
  $\invertibles{M}$.

  Two families $(u_i)_{i \in I}$ and $(v_i)_{i \in I}$ indexed by some non-empty set $I$ are equal
  \emph{up to invertibles on the left} when there is some invertible $x \in M$
  such that $\forall i \in I, u_i = xv_i$.
  % They are equal up-to-invertibles-on-the-right when they are equal
  % up-to-invertibles-on-the-left in $\dual{M}$.
\end{defi}

\begin{defi}[divisibility] An element $u$ of a monoid $M$ \emph{left-divides} a
  family $w = (w_i)_{i \in I}$ of $M$ indexed by some set $I$ when there is a
  family $(v_i)_{i \in I}$ such that $\forall i \in I, uv_i = w_i$, and we say
  that $u$ is a \emph{left-divisor} of $w$. It \emph{right-divides} it when it
  left-divides it in $\dual{M}$, and in that case $u$ is called a
  \emph{right-divisor} of $w$.

  A \emph{greatest common left-divisor} (or \emph{left-gcd}) of the family $w$
  is a left-divisor of $w$ that is left-divided by all others left-divisors of
  $w$.
  % A greatest common right-divisor (or right-gcd) of $w$ is a left-gcd of $w$
  % in $\dual{M}$.

  A family $w$ is said to be \emph{left-coprime} when it has $\varepsilon$ as a
  left-gcd, i.e. when all its left-divisors (or equivalently one of its
  left-gcds, if there is one) are right-invertible.
  % Dually it is right-coprime when all its right-divisors are right-invertible.
\end{defi}

We speak of \emph{greatest} common left-divisors because, while there may be
many such elements for a fixed family $\omega$, they all left-divide one another
and are thus equivalent in some sense.

\begin{defi}[cancellativity]
  \label{def:cancellativity}
  A monoid $M$ is said to be \emph{left-cancellative} when for all families
  $(u_i)_{i \in I}$ and $(v_i)_{i \in I}$ of $M$ indexed by some set $I$ and for
  all $w \in M$, $u = v$ as soon as $wu_i = wv_i$ for all $i \in I$. If this
  only implies that there is some $x \in \invertibles{M}$ such that $u_i = xv_i$
  for all $i \in I$, we instead say that $M$ is \emph{left-cancellative up to
    invertibles on the left}.
  % It is said to be right-cancellative when $\dual{M}$ is left-cancellative,
  % and cancellative when it is both right- and right-cancellative.

  Similarly, $M$ is said to be \emph{right-coprime-cancellative} when for all
  $u, v \in M$ and every left-coprime family $(w_i)_{i \in I}$ indexed by some set
  $I$, $u = v$ as soon as $uw_i = vw_i$ for all $i \in I$.
  % and it said to be left-coprime-cancellative when for any $u, v \in M$ and
  % left-coprime family $(w_i)_{i \in I}$ indexed by some set $I$, $\forall i
  % \in I, w_iu = w_iv$ for all left-coprime $w$ then $u = v$.
\end{defi}

\begin{defi}[n\oe{}therianity]
  \label{def:noetherianity}
  A monoid $M$ is \emph{right-n\oe{}therian} when for all sequences $(u_n)_{n \in
    \N}$ and $(v_n)_{n \in \N}$ of $M$ such that $v_n = v_{n+1}u_n$ for all $n
  \in \N$, there is some $m \in \N$ such that $u_m$ is invertible.

  In this case, we write $\rk v$ for the \emph{rank} of $v$, the (possibly
  infinite) supremum of the numbers of non-invertibles in a sequence $(u_n)_{n
    \in \N}$ that satisfies $v_n = v_{n+1}u_n$ for a sequence $(v_n)_{n \in \N}$
  with $v_0 = v$.
  % It is left-n\oe{}therian when $\dual{M}$ is right-n\oe{}therian, and n\oe{}therian
  % when both right- and right-n\oe{}therian.
\end{defi}

In other words, a monoid is right-n\oe{}therian when it has no strict infinite
chains of right-divisors (in the definition above, each $u_n \cdots u_0$
right-divides $v_0$). Note that $\rk(uv) \le \rk u + \rk v$, and if this is an
equality and the rank of every $w \in M$ is finite, $M$ is said to be
\emph{graded}.

N\oe{}therianity will be used to ensure that fixed-point algorithms terminate,
hence we will often assume that the monoids we work with satisfy some
n\oe{}therianity properties. We therefore state two additional lemmas making it
easier to work with n\oe{}therian monoids. They both assume $M$ to be
right-n\oe{}therian but the dual results also stand.

\begin{lem}
  \label{lemma:noetherianity_alternative_definition}
  $M$ is right-n\oe{}therian if and only if for all sequences $(u_n)_{n \in \N}$
  and $(v_n)_{n \in \N}$ of $M$ such that $v_n = v_{n+1}u_n$ for
  all $n \in \N$, there is some $n \in \N$ such that for all $i \ge n$, $u_i$ is
  invertible.
\end{lem}
\begin{proof}
  The reverse implication is trivial. Now assume $M$ right-n\oe{}therian, and
  consider $(u_n)_{n \in \N}$ and $(v_n)_{n \in \N}$ such that $v_n =
  v_{n+1}u_n$ for all $n \in \N$. Consider the set $I = \{ i_0 < i_1 < \ldots
  \}$ of all the indices $i$ such that $u_i$ is not invertible and, letting
  $i_{-1} = -1$, define $u_n' = u_{i_n} \cdots u_{i_{n - 1} + 1}$ and $v_n' =
  v_{i_n}$ for all $n \in \N$. Then $u_n'$ is never invertible because $u_{i_n}$
  is not but $u_{i_n - 1}, \ldots, u_{i_{n-1} + 1}$ are (by definition), yet we
  still have by induction that $v_n' = v_{n+1}'u_n'$. Since $M$ is
  right-n\oe{}therian, $I$ must be finite hence there is some $n \in \N$ such that
  for all $i \ge n$, $m_i$ is invertible.
\end{proof}

\begin{lem}
  \label{lemma:invertibles_in_noetherian_monoids}
  If $M$ is right-n\oe{}therian then all its right- and left-invertibles are
  invertible.
\end{lem}
\begin{proof}
  Assume $M$ right-n\oe{}therian and consider some $x$ that is right-invertible
  with $x_r$ its right-inverse. For $n \in \N$, set $u_{2n} = x_r$, $v_{2n} =
  \varepsilon$ and $u_{2n+1} = v_{2n+1} = x$. Then $v_n = v_{n+1}u_n$ for all $n
  \in \N$ hence by right-n\oe{}therianity $x$ is invertible. If $x$ is
  left-invertible instead, its left-inverse is right-invertible hence invertible
  and therefore so is $x$.
\end{proof}

\begin{exa}
  The canonical example of a monoid is the \emph{free monoid} $A^*$ over an
  alphabet $A$, whose elements are words with letters in $A$, whose product is
  the concatenation of words and whose unit is the empty word. Notice that the
  alphabet $A$ may be infinite. The left-divisibility relation is the prefix
  one, and the left-gcd is the longest common prefix.

  The \emph{free commutative monoid} $A^{\oast}$ over $A$ has elements the
  functions $A \rightarrow \N$ with finite support, product $(f \otimes g)(a) =
  f(a) + g(a)$ and unit the zero function $a \mapsto 0$. It is commutative ($f
  \otimes g = g \otimes f$) hence is its own dual: the divisibility relation is
  the pointwise order inherited from $\N$ and the greatest common divisor is the
  pointwise infimum.

  These two monoids are examples of \emph{trace monoids} over some $A$, defined
  as quotients of $A^*$ by commutativity relations on letters (for $A^{\oast}$,
  all the pairs of letters are required to commute, and for $A^*$ none are).
  Trace monoids have no non-trivial right- or left-invertible elements, are all
  left-cancellative, right-coprime-cancellative and right-n\oe{}therian, and the
  rank of a word is simply its number of letters.

  Another family of examples is that of \emph{groups}, monoids where all
  elements are invertible. Again, all groups are left-cancellative,
  right-coprime-cancellative and right-n\oe{}therian.

  A final monoid of interest is $(E, {\vee}, \bot)$ for $E$ any join-semilattice
  with a bottom element $\bot$. In this commutative monoid, the divisibility
  relation is the partial order on $E$, and the gcd, when it exists, is the
  infimum. This example shows that a monoid can be coprime-cancellative without
  being cancellative nor n\oe{}therian: this is for instance the case when $E =
  \R_+$ (this counter-example was pointed out to the author by Thomas
  Colcombet).
\end{exa}

\subsection{Monoidal transducers as functors}
\label{sec:category_monoidal_transducers:as_functors}

In the rest of this paper we fix a countable \emph{input alphabet} $A$ and an
\emph{output monoid} $(M, \varepsilon, \otimes)$. To differentiate between elements
of $A^*$ and elements of $M$, we write the former with Latin letters ($a, b, c,
\ldots$ for letters and $u, v, w, \ldots$ for words) and the latter with Greek
letters ($\alpha, \beta, \gamma, \ldots$ for generating elements and $\upsilon,
\nu, \omega, \ldots$ for general elements). In particular the empty word over
$A$ is denoted $e$ while the unit of $M$ is still written $\varepsilon$. We now
define our main object of study, $M$-monoidal transducers.

\begin{defi}[monoidal transducer]
  \label{def:monoidal_transducer}
  A \emph{monoidal transducer} is a tuple $(S, (\upsilon_0, s_0), t, (- \odot
  a)_{a \in A})$ where $S$ is a set of \emph{states}; $(\upsilon_0, s_0) \in M
  \times S \sqcup \{ \bot \}$ is the (possibly undefined) pair of the
  \emph{initialization value} and \emph{initial state}; $t: S \rightarrow M
  \sqcup \{ \bot \}$ is the partial \emph{termination function}; $s \odot a \in
  M \times S \sqcup \{ \bot \}$ for $a \in A$ may be undefined, and its two
  components, $- \ocircle a: S \rightarrow M \sqcup \{ \bot \}$ and $- \cdot a
 : S \rightarrow S \sqcup \{ \bot \}$, are respectively called the partial
  \emph{production function} and the partial \emph{transition function} along
  $a$.
\end{defi}

\begin{exa}
  \Cref{fig:monoidal_transducer} is a graphical representation of a monoidal
  transducer that takes its input in the alphabet $A = \{ a, b \}$ and has
  output in any monoid that is a quotient of $\Sigma^*$, with $\Sigma = \{
  \alpha, \beta \}$. Formally, it is given by $S = \{ 1, 2, 3, 4 \}$;
  $(\upsilon_0, s_0) = (\varepsilon, 1)$; $t(1) = \alpha$, $t(2) = \bot$, $t(3) =
  \alpha$ and $t(4) = \varepsilon$; and finally $1 \odot a = (\varepsilon, 2)$, $1
  \odot b = (\beta, 3)$, $3 \odot b = (\beta, 3)$ as well as $s \odot c = \bot$
  for any other $s \in S$ and $c \in A$. This transducer recognizes the function
  given by $L(b^n) = \beta^n\alpha$ (seen in the corresponding quotient monoid,
  so for instance $L(b^n) = \alpha\beta^n$ when $M = \Sigma^{\oast}$) for all $n
  \in \N$ and $L(w) = \bot$ otherwise. Other examples are given by
  \Cref{fig:transducers}.

  \begin{figure}[ht]
    \centering
    \includefigure{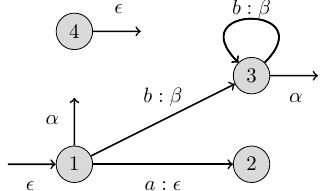}
    \caption{A monoidal transducer $\B$}
    \label{fig:monoidal_transducer}
  \end{figure}
\end{exa}

To apply the framework of \Cref{sec:categorical_framework} we first need to
model monoidal transducers as functors. We thus design a tailored output
category that in particular matches the one for classical transducers when $M$
is a free monoid \cite[Section 4]{colcombetAutomataMinimizationFunctorial2020}.
We write $\T_M$ for the monad on $\mathbf{Set}$ given by $\T_M X = M \times X +
1 = (M \times X) \sqcup \{ \bot \}$ (in Haskell, this monad is the composite of
the \texttt{Maybe} monad and a \texttt{Writer} monad). Its unit $\eta: \Id
\Rightarrow \T_M$ is given by $\eta_X(x) = (\varepsilon, x)$ and its
multiplication $\mu: \T_M^2 \Rightarrow \T_M$ is given by $\mu_X((\upsilon,
(\nu, x))) = (\nu\upsilon, x)$ (composition in $M$ is read from left to right so
the outermost output value is appended to the innermost one), $\mu_X((\upsilon,
\bot)) = \bot$ and $\mu_X(\bot) = \bot$. Recall that the Kleisli category
$\Kl(\T_M)$ for the monad $\T_M$ has sets for objects and arrows $X \klarrow Y$
(notice the different symbol) those functions $f^\dagger: \T_MX \rightarrow \T_M
Y$ such that $f^\dagger(\bot) = \bot$, $f^\dagger(\upsilon, x) = (\upsilon\nu,
y)$ when $f^\dagger(\varepsilon, x) = (\nu, y)$ and $f^\dagger(\upsilon, x) =
\bot$ when $f(\varepsilon, x) = \bot$: in particular, such an arrow is entirely
determined by its restriction $f: X \rightarrow \T_M Y$, and we will freely
switch between these two points of view for the sake of conciseness. The
identity on $X$ is then given by the identity function $\id_{\T_M X} =
\eta_X^\dagger: \T_M X \rightarrow \T_M X$, and the composition of two arrows $X
\klarrow Y \klarrow Z$ is given by the composition of the underlying functions
$\T_M X \rightarrow \T_M Y \rightarrow \T_M Z$.

$M$-transducers are in one-to-one correspondance with
$(\Kl(\T_M),1,1)$-automata, i.e. functors $\A: \I \rightarrow \Kl(\T_M)$ such
that $\A(\In) = \A(\Out) = 1$: $(S, (\upsilon_0, s_0), t, (- \odot a)_{a \in
  A})$ is modelled by the functor $\A: \I \rightarrow \Kl(\T_M)$ given by
$\A(\St) = S$, $\A(\triangleright) = * \mapsto (\upsilon_0, s_0): 1 \rightarrow
M \times S + 1$, $\A(w) = s \mapsto s \odot w = (((s \odot a_1) \odot^\dagger
a_2) \odot^\dagger \cdots) \odot^\dagger a_n: S \rightarrow M \times S + 1$ for
$w = a_1 \cdots a_n \in A^*$, and $\A(\triangleleft) = t: S \rightarrow M + 1
\cong M \times 1 + 1$.

\begin{defi}
  We write $\Trans_M$ for the category of $(\Kl(\T_M),1,1)$-automata: the
  objects are $M$-transducers seen as functors and the morphisms are natural
  transformations between them. Given a $(\Kl(\T_M),1,1)$-language $\L: \O
  \rightarrow \Kl(\T_M)$, we write $\Trans_M(\L)$ for the subcategory of
  $M$-transducers $\A: \I \rightarrow \Kl(\T_M)$ that recognize $\L$, i.e. such
  that $\A \circ \iota = \L$.
\end{defi}

Under this correspondance, the language recognized by a transducer
$(S,(\upsilon_0,s_0),t,\odot)$ is thus a function $L: A^* \rightarrow M + 1$
given by $L(w) = t^\dagger((\upsilon_0, s_0) \odot^\dagger w)$, and a morphism
between two transducers $(S_1, (\upsilon_1, s_1), t_1, {\odot_1})$ and $(S_2,
(\upsilon_2, s_2), t_2, {\odot_2})$ is a function $f: S_1 \rightarrow M \times
S_2 + 1$ such that $f^\dagger(\upsilon_1, s_1) = (\upsilon_2, s_2)$, $t_1(s) =
t_2^\dagger(f(s))$ and $f^\dagger(s \odot a) = f^\dagger(s) \odot^\dagger a$.

When $M = B^*$ for some alphabet $B$, $M$-transducers coincide with the
classical notion of deterministic one-way transducers and the minimal transducer
is given by \Cref{def:minimal_object}
\cite{choffrutMinimizingSubsequentialTransducers2003,
  colcombetAutomataMinimizationFunctorial2020}. To study the notion of minimal
monoidal transducer, it is thus natural to try to follow this framework as well.

\subsection{The initial and final monoidal transducers recognizing a function}
\label{sec:category_monoidal_transducers:initial_final_objects}

To apply the framework of \Cref{sec:categorical_framework} to
$(\Kl(\T_M),1,1)$-automata, we need three ingredients in $\Kl(\T_M)$: countable
copowers of $1$, countable powers of $1$, and a factorization system.

We start with the first ingredient, countable copowers of $1$. Since $\Set$ has
arbitrary coproducts, $\Kl(\T_M)$ has arbitrary coproducts as well as any
Kleisli category for a monad over a category with coproducts does
\cite[Proposition 2.2]{szigetiLimitsColimitsKleisli1983}. Hence
\Cref{thm:minimal_automaton} applies:

\begin{cor}[initial transducer]
  \label{corollary:initial_transducer}
  For any $(\Kl(\T_M), 1, 1)$-language $\L$, $\Trans_M(\L)$ has an initial
  object $\A^{init}(\L)$ with state-set $S^{init} = A^*$, initial state
  $s_0^{init} = e$, initialization value $\upsilon_0^{init} = \varepsilon$,
  termination function $t^{init}(w) = \L(\word{w})(*)$ and transition function
  $w \odot^{init} a = (\varepsilon, wa)$. Given any other transducer $\A = (S,
  (\upsilon_0, s_0), t, {\odot})$ recognizing $\L$, the unique transducer
  morphism $f: \A^{init}(\L) \Longrightarrow \A$ is given by the function $f:
  A^* \rightarrow M \times S + 1$ such that $f(w) = \A(\triangleright w)(*) =
  (\upsilon_0, s_0) \odot^\dagger w$.
\end{cor}

Similarly, to get a final transducer in $\Trans_M(\L)$ for some $\L$,
\Cref{thm:minimal_automaton} tells us that it is enough for $\Kl(\T_M)$ to have
all countable powers of $1$. This holds for classical transducers, when $M$ is a
free monoid \cite[Lemma 4.7]{colcombetAutomataMinimizationFunctorial2020}. Hence
we study conditions on the monoid $M$ for $\Kl(\T_M)$ to have these powers.

To this means, given a countable set $I$ we consider partial functions $\Lambda
: I \rightarrow M + 1$. We write $\bot^I$ for the nowhere defined function $i
\mapsto \bot$ and $(M + 1)_*^I = (M + 1)^I - \{ \bot^I \}$ for the set of
partial functions that are defined somewhere. If $I \subseteq J$, $(M + 1)_*^I$
may thus be identified with the subset of partial functions of $(M + 1)_*^J$
that are undefined on $J - I$. We extend the product ${\otimes}: M^2 \rightarrow
M$ of $M$ to a function $M \times (M + 1)_*^I \rightarrow (M+1)_*^I$ by setting
$(\upsilon \otimes \Lambda)(i) = \upsilon \otimes \Lambda(i)$ for $i \in I$ such
that $\Lambda(i) \neq \bot$ and $(\upsilon \otimes \Lambda)(i) = \bot$
otherwise. The universal property of the (categorical) product in $\Kl(\T_M)$
then translates as:

\begin{prop}
  \label{prop:countable_products_lgcd_red}
  The following are equivalent:
  \begin{enumerate}[label=(\arabic*), series=countable_products_lgcd_red]
  \item \label{prop:countable_products_lgcd_red:item:powers_1}
    $\Kl(\T_M)$ has all countable powers of $1$;
  \item \label{prop:countable_products_lgcd_red:item:lgcd_red}
    there are two functions $\lgcd: (M + 1)_*^{\N} \rightarrow M$ and $\red
   : (M + 1)_*^{\N} \rightarrow (M + 1)_*^{\N}$ such that
    \begin{enumerate}
    \item \label{prop:countable_products_lgcd_red:item:lgcd_red:item:sound} for all $\Lambda \in (M + 1)_*^{\N}$, $\Lambda =
      \lgcd(\Lambda) \red(\Lambda)$;
    \item \label{prop:countable_products_lgcd_red:item:lgcd_red:item:unique} for
      all $\Gamma, \Lambda \in (M + 1)_*^{\N}$ and $\upsilon, \nu \in M$, if
      $\upsilon \red(\Gamma) = \nu \red(\Lambda)$ then $\upsilon = \nu$ and
      $\red\Gamma = \red\Lambda$;
    \end{enumerate}
  \item \label{prop:countable_products_lgcd_red:item:products}
    $\Kl(\T_M)$ has all countable products.
  \end{enumerate}
  Moreover when these hold, since any countable set $I$ embeds into $\N$,
  $\lgcd$ and $\red$ can be extended to $(M+1)_*^I$. $\lgcd\Lambda$ is then a
  left-gcd of $(\Lambda(i))_{i \mid \Lambda(i) \neq \bot}$ and the product of
  $(X_i)_{i \in I}$ for some $I \subseteq \N$ is the set of pairs $(\Lambda,
  (x_i)_{i \in I})$ such that $\Lambda \in \red((M+1)_*^I)$ and, for all $i \in
  I$, $x_i \in X_i$ if $\Lambda(i) \neq \bot$ and $x_i = \bot$ otherwise. In
  particular, the $I$-th power of $1$ is the set of \emph{irreducible} partial
  functions $I \rightarrow M + 1$:
  \[ \prod_I 1 = \Irr(I,M) = \{ \red\Lambda \in (M + 1)_*^I \mid \Lambda \in (M + 1)_*^I
    \} \]
\end{prop}
\begin{proof}
  $\ref{prop:countable_products_lgcd_red:item:products} \Rightarrow
  \ref{prop:countable_products_lgcd_red:item:powers_1}$ holds by definition.

  Let us now show $\ref{prop:countable_products_lgcd_red:item:powers_1}
  \Rightarrow \ref{prop:countable_products_lgcd_red:item:lgcd_red}$. Assume
  $\Pi_{\N} 1$ exists in $\Kl(\T_M)$, and fix some $\Lambda \in (M + 1)_*^{\N}$.

  By the universal property of the product, the cone $(\Lambda(n): 1 \rightarrow
  M + 1)_{n \in \N}$ factors through some unique $h: 1 \rightarrow M \times
  (\Pi_{\N}1) + 1$. Write $h(*) = (\delta, x_{\Lambda})$ so that $\delta \in M$
  and $x_{\Lambda} \in \Pi_{\N} 1$ -- the case $h(*) = \bot$ is excluded as that
  would imply $\Lambda = \bot^{\N}$. We then set $\lgcd \Lambda = \delta$ and
  $\red(\Lambda)(n) = \pi_n(x_{\Lambda})$ for all $n \in \N$, so that in
  particular $\red(\Lambda) \neq \bot^{\N}$. Since $\Lambda(n) = \pi_n \circ h$
  for all $n \in \N$ we get that $\Lambda = \lgcd(\Lambda) \red(\Lambda)$, and
  if $\Xi = \upsilon \red(\Gamma) = \nu \red(\Lambda)$, then $(* \mapsto
  \Xi(n))_{n \in \N}$ factors through both $g(*) = (\upsilon, x_{\Gamma})$ and
  $h(*) = (\nu, x_{\Lambda})$: $g = h$ hence $\upsilon = \nu$ and $\red \Gamma =
  \red \Lambda$. $\lgcd$ and $\red$ thus satisfy conditions
  \ref{prop:countable_products_lgcd_red:item:lgcd_red:item:sound} and
  \ref{prop:countable_products_lgcd_red:item:lgcd_red:item:unique}.

  In particular, when these conditions are satisfied $\lgcd(\Lambda)
  \red(\Lambda) = \Lambda$ hence $\lgcd \Lambda$ left-divides $(\Lambda(i))_{i
    \mid \Lambda(i) \neq \bot}$. If $\delta$ left-divides this family then
  $\Lambda = \delta \Gamma$ for some $\Gamma \neq \bot^{\N}$ hence $\Lambda =
  (\delta \lgcd(\Gamma)) \red(\Gamma)$ and thus $\delta$ left-divides
  $\lgcd(\Lambda) = \delta \lgcd(\Gamma)$.

  Finally, let us show $\ref{prop:countable_products_lgcd_red:item:lgcd_red}
  \Rightarrow \ref{prop:countable_products_lgcd_red:item:products}$ along with
  the formula for the product of a countable number of objects. Given $(X_i)_{i
    \in I}$ indexed by some $I \subseteq \N$, define $\prod_I X_i$ as in the
  statement of the proposition and the projection $\pi_j: \prod_I X_i \klarrow
  X_j$ by $\pi_j(\Lambda, (x_i)_{i \in I}) = (\Lambda(j), x_j)$ if $\Lambda(j)
  \neq \bot$ and $\pi_j(\Lambda, (x_i)_{i \in I}) = \bot$ otherwise. Given a
  cone $(f_i: 1 \klarrow X_i)$ and some $i \in I$, write $f_i(*) = (\Lambda(i),
  x_i)$ when $f_i \neq \bot$ and $\Lambda(i) = x_i = \bot$ otherwise. Define now
  $h: 1 \klarrow \prod_I X_i$ by $h(*) = (\lgcd\Lambda, (\red\Lambda, (x_i)_{i
    \in I}))$ if $\Lambda \neq \bot^\N$ and $h(*) = \bot$ otherwise. We
  immediately have that $f_i = \pi_i \circ h$ by condition
  \ref{prop:countable_products_lgcd_red:item:lgcd_red:item:sound}, and that $h$
  is the only such function by condition
  \ref{prop:countable_products_lgcd_red:item:lgcd_red:item:unique}.
\end{proof}

We also write $\lgcd(\bot^{\N}) = \bot$, $\red(\bot^{\N}) = \bot^\N = \bot$ and
do not distinguish between $(\bot, \bot)$ and $\bot$. Also note that:

\begin{lem}
  Conditions \ref{prop:countable_products_lgcd_red:item:lgcd_red:item:sound} and
  \ref{prop:countable_products_lgcd_red:item:lgcd_red:item:unique} are
  equivalent to saying that
  \begin{enumerate}[label=(\arabic*), resume=countable_products_lgcd_red]
  \item \label{prop:countable_products_lgcd_red:item:lgcd_red_bis}
    \begin{enumerate}
    \item \label{prop:countable_products_lgcd_red:item:lgcd_red_bis:item:injectivity}
      $\langle \lgcd, \red \rangle$ is injective;
    \item \label{prop:countable_products_lgcd_red:item:lgcd_red_bis:item:red_irreducible}
      for all $\Lambda \in (M + 1)_*^{\N}$, $\lgcd(\red \Lambda) = \varepsilon$
      and $\red(\red \Lambda) = \red \Lambda$;
    \item \label{prop:countable_products_lgcd_red:item:lgcd_red_bis:morphisms}
      for all $\Lambda \in (M + 1)_*^{\N}$ and $\upsilon \in M$, $\lgcd(\upsilon
      \Lambda) = \upsilon \lgcd(\Lambda)$ and $\red(\upsilon \Lambda) =
      \red(\Lambda)$.
    \end{enumerate}
  \end{enumerate}
 \end{lem}
\begin{proof}
  Assuming conditions
  \ref{prop:countable_products_lgcd_red:item:lgcd_red:item:sound} and
  \ref{prop:countable_products_lgcd_red:item:lgcd_red:item:unique}, we get that
  \begin{itemize}[align=left]% I removed the itemize dots because the enumeration labels act as bullet points. If you do not agree with this, let us know!
  \item[
    \ref{prop:countable_products_lgcd_red:item:lgcd_red_bis:item:injectivity}:]
    if $\lgcd \Gamma = \lgcd \Lambda$ and $\red \Gamma = \red \Lambda$ then
    \[
      \Gamma = \lgcd(\Gamma) \red(\Gamma) = \lgcd(\Lambda) \red(\Lambda) =
      \Lambda
             \]
  \item[
    \ref{prop:countable_products_lgcd_red:item:lgcd_red_bis:item:red_irreducible}:]
    $\red \Lambda = \lgcd(\red \Lambda) \red(\red \Lambda)$ hence $\lgcd(\red
    \Lambda) = \varepsilon$ and $\red(\red \Lambda) = \red \Lambda$;
  \item[ \ref{prop:countable_products_lgcd_red:item:lgcd_red_bis:morphisms}:]
    $(\upsilon \lgcd(\Lambda)) \red(\Lambda) = \upsilon \Lambda$ hence $\upsilon
    \lgcd(\Lambda) = \lgcd(\upsilon \Lambda)$ and $\red(u \Lambda) = \red
    \Lambda$.
  \end{itemize}

  And conversely, assuming conditions
  \ref{prop:countable_products_lgcd_red:item:lgcd_red_bis:item:injectivity},
  \ref{prop:countable_products_lgcd_red:item:lgcd_red_bis:item:red_irreducible}
  and \ref{prop:countable_products_lgcd_red:item:lgcd_red_bis:morphisms}, we get
  that
  \begin{itemize}[align=left]
  \item[ \ref{prop:countable_products_lgcd_red:item:lgcd_red:item:sound}:] by
    injectivity, $\Lambda = \lgcd(\Lambda) \red(\Lambda)$ since
    \[ \lgcd(\lgcd(\Lambda) \red(\Lambda)) = \lgcd(\Lambda) \lgcd(\red\Lambda) = \lgcd\Lambda \]
    and
    \[ \red(\lgcd(\Lambda) \red(\Lambda)) = \red(\red\Lambda) = \red\Lambda \]
  \item[ \ref{prop:countable_products_lgcd_red:item:lgcd_red:item:unique}:] if
    $\upsilon \red(\Gamma) = \nu \red(\Lambda)$ then
    \[ \upsilon = \upsilon \lgcd(\red \Gamma) = \lgcd(\upsilon \red(\Gamma)) =
      \lgcd(\nu \red(\Lambda)) = \nu \lgcd(\red \Lambda) = \nu \] and
    \[ \red \Gamma = \red(\red \Gamma) = \red(\upsilon \red(\Gamma)) = \red(\nu
      \red(\Lambda)) = \red(\red \Lambda) = \red \Lambda \qedhere \]
  \end{itemize}
\end{proof}

\begin{cor}
  When the functions $\lgcd$ and $\red$ exist, the final transducer
  $\A^{final}(\L)$ recognizing a $(\Kl(\T_M), 1, 1)$-language $\L$ exists and
  has state-set $S^{final} = \Irr(A^*,M)$, initial state $s_0^{final} = \red
  \L$, initialization value $\upsilon_0^{final} = \lgcd \L$, termination
  function $t^{final}(\Lambda) = \Lambda(e)$ and transition function $\Lambda
  \odot^{final} a = (\lgcd(\inv{a}\Lambda), \red(\inv{a}\Lambda))$ where we
  write $(\inv{a}\Lambda)(w) = \Lambda(aw)$ for $a \in A$. Given any other
  transducer $\A = (S, (\upsilon_0, s_0), t, {\odot})$ recognizing $\L$, the
  unique transducer morphism $f: \A \Longrightarrow \A^{final}(\L)$ is given by
  the function $f: S \rightarrow M \times \Irr(A^*,M) + 1$ such that $f(s) =
  (\lgcd \L_s, \red \L_s)$ where $\L_s(\triangleright w \triangleleft)(*) = \A(w
  \triangleleft)(s)$ is the function recognized by $\A$ from the state $s$.
\end{cor}

In practice we will assume that $M$ is right-n\oe{}therian to ensure algorithms
terminate. It is thus interesting to see what the existence of the powers of $1$
implies in this specific case (and in particular, by
\Cref{lemma:invertibles_in_noetherian_monoids}, when $M$ is such that all right-
and left-invertibles are invertibles).

\begin{lem}
  \label{lemma:countable_powers_left_right_invertibles}
  If right- and left-invertibles of $M$ are all invertibles and if
  $\Kl(\T_M)$ has all countable powers of $1$ then $M$ is both
  left-cancellative up to invertibles on the left and
  right-coprime-cancellative, and all non-empty countable families of $M$ have a
  unique left-gcd up to invertibles on the right.
\end{lem}
\begin{proof}
  Let us first show left-cancellativity up to invertibles. If $\upsilon \gamma_i
  = \upsilon \lambda_i$ for some $\upsilon \in M$ and two countable families
  $(\gamma_i)_{i \in I}, (\lambda_i)_{i \in I}$ of elements of $M$ with $I
  \subseteq \N_{>0}$, then defining $\Gamma(0) = \Lambda(0) = \varepsilon$,
  $\Gamma(i) = \gamma_i$ and $\Lambda(i) = \lambda_i$ for all $i \in I$ and
  $\Gamma(j) = \Lambda(j) = \bot$ for all $j \notin I$, we have that $\upsilon
  \Gamma = \upsilon \Lambda$. Hence $\red \Gamma = \red \Lambda$ and $\upsilon
  \lgcd(\Gamma) = \upsilon \lgcd(\Lambda)$. But $\lgcd \Gamma$ and $\lgcd
  \Lambda$ left-divide $\varepsilon = \Gamma(0) = \Lambda(0)$ hence they are
  right-invertible and thus invertible: $\gamma_i$ and $\lambda_i =
  \lgcd(\Lambda) \inv{\lgcd(\Gamma)} \gamma_i$ are equal up to an invertible
  (that does not depend on $i$) on\linebreak the left.

  Moreover, if $\upsilon \Lambda = \nu \Lambda$ for some $\Lambda$ such that
  $(\Lambda(i))_{i \mid \Lambda(i) \neq \bot}$ is left-coprime, then $\lgcd
  \Lambda$ is right-invertible (by definition of left-coprimality) and since
  \[ \upsilon \lgcd(\Lambda) = \lgcd(\upsilon \Lambda) = \lgcd(\nu \Lambda) =
    \nu \lgcd(\Lambda) \] we get by left-invertibility of $\lgcd \Lambda$ that
  $\upsilon = \nu$.

  Finally, consider $\delta$ a left-gcd of some non-empty countable family
  encoded as $\Lambda \in (M + 1)_*^{\N}$. Then there is some $\upsilon$ such
  that $\delta = \lgcd(\Lambda) \upsilon$ (since $\lgcd(\Lambda)$ left-divides
  $\Lambda$) and some $\Gamma$ such that $\delta \Gamma = \Lambda$. Hence
  \[ \lgcd \Lambda= \lgcd(\delta \Gamma) = \delta \lgcd(\Gamma) = \lgcd(\Lambda)
    \upsilon \lgcd(\Gamma) \] By left-cancellativity up to invertibles,
  $\upsilon \lgcd(\Gamma) \in \invertibles{M}$ so $\upsilon$ is right-invertible
  hence invertible: $\delta = \lgcd \Lambda$ up to invertibles on the right.
\end{proof}

And conversely, these conditions on $M$ are enough for $\Kl(\T_M)$ to have all
countable powers of $1$ (even when $M$ is not n\oe{}therian), while being easier to
show than properly defining the two functions $\lgcd$ and $\red$.

\begin{lem}
  \label{lemma:countable_powers_sufficient_conditions}
  If $M$ is both left-cancellative up to invertibles on the left and
  right-coprime-cancellative, and all non-empty countable subsets of $M$ have a unique
  left-gcd up to invertibles on the right, then $\Kl(\T_M)$ has all
  countable powers of $1$.
\end{lem}
\begin{proof}
  Split the set of those $\Lambda \in (M + 1)_*^{\N}$ such that $(\Lambda(i))_{i
    \mid \Lambda(i) \neq \bot}$ is left-coprime into the equivalence classes
  given by $\chi \Lambda \sim \Lambda$ for all $\chi \in \invertibles{M}$. Then,
  for each equivalence class $C$ pick a $\red C$ in $C$ (using the axiom of
  choice) and for all $\upsilon \in M$ define $\lgcd(\upsilon \red(C)) =
  \upsilon$ and $\red(\upsilon \red(C)) = \red C$ so that in particular
  $\lgcd(\red C) = \varepsilon$ and $\red(\red C) = \red C$.

  This is well-defined because if $\upsilon \red(C) = \nu \red(D)$ for
  $\upsilon, \nu \in M$ and two equivalence classes $C, D$, then if $\delta$ is
  a left-gcd of $\upsilon \red(C)$ we have that $\delta = \upsilon \upsilon'$
  for some $\upsilon'$ (since $\upsilon$ left-divides $\upsilon \red(C)$) and
  there is some $\Lambda$ such that $\upsilon \red(C) = \upsilon \upsilon'
  \Lambda$. But then by left-cancellativity up to invertibles, there is some
  $\chi \in \invertibles{M}$ such that $\upsilon' \Lambda = \chi \red(C)$, hence
  $\inv{\chi} \upsilon'$ left-divides $\red C$ and as such is right-invertible
  (it left-divides $\varepsilon$, a left-gcd of $\red C$), making $\upsilon'$
  right-invertible as well. This shows that $\upsilon$ is a left-gcd of
  $\upsilon \red(C) = \nu \red(D)$ and we show similarly that this is also true
  of $\nu$, hence by unicity of the left-gcd there is a $\xi \in
  \invertibles{M}$ such that $\upsilon = \nu \xi$. Therefore by
  left-cancellativity up to invertibles there is another invertible $\xi' \in
  \invertibles{M}$ such that $\xi \red(C) = \xi' \red(D)$ hence by definition $C
  = D$ and, by right-coprime-cancellativity, $\upsilon = \nu$.

  Moreover this defines $\lgcd \Lambda$ and $\red \Lambda$ for all $\Lambda$
  because if $\delta$ is a left-gcd of $\Lambda$, then $\Lambda = \delta \Gamma$
  for some left-coprime $\Gamma$ (if $\delta'$ left-divides $\Gamma$ then
  $\delta \delta'$ left-divides $\Lambda$ hence $\delta$, therefore $\delta'$ is
  invertible by left-cancellativity up to invertibles) hence $\Lambda = \delta
  \chi \red(C)$ if $\Gamma \in C$ and $\Gamma = \chi \red(C)$.

  We have thus defined two functions $\lgcd$ and $\red$ that immediately satisfy
  the conditions \ref{prop:countable_products_lgcd_red:item:lgcd_red:item:sound}
  and \ref{prop:countable_products_lgcd_red:item:lgcd_red:item:unique} of
  \Cref{prop:countable_products_lgcd_red}.
\end{proof}

\begin{exa}
  When $M$ is a group it is cancellative (because all elements are invertible)
  and all countable families have a unique left-gcd up to invertibles on the
  right ($\varepsilon$ itself) hence
  \Cref{lemma:countable_powers_sufficient_conditions} applies and $\Trans_M(\L)$
  always has a final object.

  The same is true when $M$ is a trace monoid (the left-gcd then being the
  longest common prefix, whose existence is guaranteed by
  \cite[Proposition~1.3]{coriAutomatesCommutationsPartielles1985}).

  Conversely, the monoids given by join semi-lattices are not left-cancellative
  up to invertibles in general. In $\R_+$ for instance, there are ways to define
  the functions $\lgcd$ and $\red$ but they may not satisfy condition
  \ref{prop:countable_products_lgcd_red:item:lgcd_red_bis:morphisms}, more
  precisely that $\red(\upsilon \Lambda) = \red \Lambda$. This is expected, as
  there may be several non-isomorphic ways to minimize automata with outputs in
  these monoids, which is incompatible with the framework of
  \Cref{def:minimal_object}.
\end{exa}

\Cref{lemma:countable_powers_sufficient_conditions} provides sufficient
conditions that are reminiscent of those developed in
\cite{gerdjikovGeneralClassMonoids2018} for the minimization of monoidal
transducers. These conditions are stronger than ours but still similar: the
output monoid is assumed to be both left- and right-cancellative (the LC and RC
axioms in \emph{ibid.}), which in particular implies the unicity up to
invertibles on the right of the left-gcd whose existence is also assumed. They
do only require the existence of left-gcds for finite families (the LSL axiom in
\emph{ibid.}) (whereas we ask for left-gcds of countable families), which would
not be enough for our sake since the categorical framework also encompasses the
existence of minimal (infinite) automata for non-regular languages, but in
practice our algorithms will only use binary left-gcds as well. We conjecture
that, when only those binary left-gcds exist, the existence of a unique minimal
transducer is explained categorically by the existence of a final transducer in
the category of transducers whose states all recognize functions that are
themselves recognized by finite transducers. Where the two sets of conditions
really differ is in the conditions required for the termination of the
algorithms: where we will require right-n\oe{}therianity of $M$, they require that
if some $\nu$ left-divides both some $\omega$ and $\upsilon \omega$ for some
$\upsilon$, then $\nu$ should also left-divide $\upsilon \nu$ (the GCLF axiom in
\emph{ibid.}). This last condition leads to better complexity bounds than
right-n\oe{}therianity, but misses some otherwise simple monoids that satisfy
right-n\oe{}therianity, e.g. $\{\alpha,\beta\}^*$ but where we also let $\alpha$
and $\beta^2$ commute. Conversely, Gerdjikov's main non-trivial example, the
tropical monoid $(\R_+,0,+)$, is not right-n\oe{}therian. It can still be dealt
with in our context by considering submonoids (finitely) generated by the output
values of a finite transducer's transitions, these monoids themselves being
right-n\oe{}therian.

\subsection{Factorization systems}
\label{sec:category_monoidal_transducers:factorization_systems}

The last ingredient we need in order to be able to apply the framework of
\Cref{sec:categorical_framework} is a factorization
system on $\Trans_M(\L)$. By \Cref{lemma:factorization_system_auto}, it is
enough to find a factorization system on $\Kl(\T_M)$.

When $M$ is a free monoid, define $\E = \Surj$ to be the class of those
functions $f: X \rightarrow M \times Y + 1$ that are surjective on $Y$ and $\M
= \Inj \cap \Eps \cap \Tot$ to be the class of those functions $f: X
\rightarrow M \times Y + 1$ that are total ($f \in \Tot$), injective when
corestricted to $Y$ ($f \in \Inj$), and only produce the empty word ($f \in
\Eps$). Then the $(\E, \M)$-minimal transducer recognizing a function is the one
defined by Choffrut \cite{choffrutMinimizingSubsequentialTransducers2003,
  colcombetAutomataMinimizationFunctorial2020}: in particular, the fact that the
minimal transducer $(\E, \M)$-divides all other transducers means (thanks to the
surjectivity of $\E$-morphisms and the injectivity of $\M$-morphisms) that it
has the smallest possible state-set and produces its outputs as early as
possible.

It is thus natural to try and extend this factorization system to $\Kl(\T_M)$
for arbitrary $M$. It is not enough by itself because isomorphisms may produce
invertible elements that may be different from $\varepsilon$: $\Iso \subsetneq
\Surj \cap \Inj \cap \Eps \cap \Tot$, yet we need the intersection $\E \cap \M$
to be $\Iso$. $\M$-morphisms must thus be able to produce invertible elements as
well. Formally, define therefore $\Surj$, $\Inj$, $\Tot$ and $\Inv$ as follows.
For $f: X \rightarrow M \times Y + 1$, write $f_1: X \rightarrow M + 1$ for
its projection on $M$ and $f_2: X \rightarrow Y + 1$ for its projection on $Y$:
we let $f \in \Surj$ whenever $f_2$ is surjective on $Y$, $f \in \Inj$ whenever
$f_2$ is injective when corestricted to $Y$, $f \in \Tot$ whenever $f(x) \neq
\bot$ for all $x \in X$ and $f \in \Inv$ whenever $f_1(x)$ is either $\bot$ or
in $\invertibles{M}$. The point is that when replacing $\Eps$ with $\Inv$, we
get back that

\begin{lem}
  \label{lemma:four_classes}
  In $\Kl(\T_M)$, $\Iso = \Surj \cap \Inj \cap \Inv \cap \Tot$, and these four
  classes are all closed under composition (within themselves).
\end{lem}
\begin{proof}
  Closure under composition is immediate.

  If $f: X \rightarrow M \times Y + 1$ is in $\Surj \cap \Inj \cap \Inv \cap
  \Tot$, it can be restricted to a function $\langle f_1, f_2 \rangle: X
  \rightarrow M \times Y$ ($f \in \Tot$). $f_2$ must be bijective ($f \in \Surj
  \cap \Inj$) and $f_1(X) \subseteq \invertibles{M}$ ($f \in \Inv$). Hence $f$
  has an inverse, given by $(\inv{f})(y) = \left((\inv{f_1(\inv{f_2}(y))},
    \inv{f_2}(y)\right)$.

  Conversely, if $f: X \rightarrow M \times Y + 1$ is an isomorphism it has an
  inverse $\inv{f}: Y \rightarrow M \times X + 1$ such that $f \circ \inv{f} =
  \id_Y$ and $\inv{f} \circ f = \id_X$. For all $x \in X$,
  $(\inv{f})^\dagger(f(x)) = (\varepsilon, x)$ hence $f(x) \neq \bot$: $f \in
  \Tot$, and similarly for $\inv{f}$. Writing $f = \langle f_1, f_2 \rangle$ and
  $\inv{f} = \langle (\inv{f})_1, (\inv{f})_2 \rangle$, we have that $f_2$ is a
  bijection with inverse $\inv{f_2} = (\inv{f})_2$, hence $f \in \Surj \cap
  \Inj$. Finally, for all $x \in X$ we have that $f_1(x)(\inv{f})_1(f_2(x)) =
  \varepsilon$ and conversely, hence $f_1(x)$ is invertible and $f \in \Inv$.
\end{proof}

The different ways in which we may distribute these four classes into the two classes
$\E$ and $\M$ leads to not just one but three interesting factorization systems:

\begin{prop}
  \label{prop:factorization_systems}
  $(\E_1,\M_1) = (\Surj \cap \Inj \cap \Inv, \Tot)$, $(\E_2, \M_2) = (\Surj \cap
  \Inj, \Inv \cap \Tot)$ and $(\E_3, M_3) = (\Surj, \Inj \cap \Inv \cap \Tot)$
  are all factorization systems in $\Kl(\T_M)$.
\end{prop}
\begin{proof}
  By \Cref{lemma:four_classes} we only need to show for each $i \in \{ 1,2,3 \}$
  that every $f: X \klarrow Y$ may be factored as $m \circ e$ with $e \in
  \E_i$ and $m \in \M_i$, and that $\E_i$ and $\M_i$ satisfy the diagonal
  fill-in property of \Cref{def:factorization_system}.

  \begin{itemize}
  \item $(\E_1, \M_1)$. Any $f: X \rightarrow M \times Y + 1$ may be factored
    through $e: X \rightarrow M \times f^{-1}(M \times Y) + 1$ (given by $e(x) =
    (\varepsilon, x)$ if $f(x) \neq \bot$ and $e(x) = \bot$ otherwise) and $m:
    f^{-1}(M \times Y) \rightarrow M \times Y + 1$ (given by $m(x) = f(x) \neq \bot$).

    Moreover, given a commuting diagram
    \begin{center}
      \begin{tikzcd}
        X \arrow[d, kleisli, "u"'] \arrow[r, "e", kleisli, two heads] & Y_1
        \arrow[d, kleisli, "v"] \\
        Y_2 \arrow[r, "m", kleisli, tail] & Z
      \end{tikzcd}
    \end{center}
    the only possible choice for a $\phi: Y_1 \rightarrow M \times Y_2 + 1$ is
    given by $\phi(y_1) = \bot$ if $v(y_1) = \bot$ and $\phi(y_1) =
    (\inv{\upsilon_e}\upsilon_u, y_2)$ if $e(x) = (\upsilon_e, y_1)$, $u(x) =
    (\upsilon_u, y_2)$, $v(y_1) = (\upsilon_v, z)$, $m(y_2) = (\upsilon_m, z)$ and
    $\upsilon_e\upsilon_v = \upsilon_u\upsilon_m$. This definition does not depend
    on the choice of $x$ because of the bijectiveness property of $e$, and we
    immediately have the commuting diagram
    \begin{center}
      \begin{tikzcd}
        X \arrow[d, kleisli, "u"'] \arrow[r, "e", kleisli, two heads] & Y_1
        \arrow[d,
        kleisli, "v"] \arrow[ld, kleisli, "\phi"'] \\
        Y_2 \arrow[r, "m", kleisli, tail] & Z
      \end{tikzcd}
    \end{center}

  \item $(\E_2, \M_2)$. Any $f: X \rightarrow M \times Y + 1$ may be factored
    through $e: X \rightarrow M \times f^{-1}(M \times Y) + 1$ (given by $e(x) =
    (\upsilon, x)$ if $f(x) = (\upsilon, \--)$ and $e(x) = \bot$ otherwise) and $m
   : f^{-1}(M \times Y) \rightarrow M \times Y + 1$ (given by $m(x) = (\varepsilon,
    y)$ when $f(x) = (\--, y)$).

    Finally, given a commuting diagram
    \begin{center}
      \begin{tikzcd}
        X \arrow[d, kleisli, "u"'] \arrow[r, "e", kleisli, two heads] & Y_1
        \arrow[d, kleisli, "v"] \\
        Y_2 \arrow[r, "m", kleisli, tail] & Z
      \end{tikzcd}
    \end{center}
    the only possible choice for a $\phi: Y_1 \rightarrow M \times Y_2 + 1$ is
    given by $\phi(y_1) = \bot$ if $v(y_1) = \bot$ and $\phi(y_1) =
    (\upsilon_v\inv{\upsilon_m}, y_2)$ if $e(x) = (\upsilon_e, y_1)$, $u(x) =
    (\upsilon_u, y_2)$, $v(y_1) = (\upsilon_v, z)$, $m(y_2) = (\upsilon_m, z)$ and
    $\upsilon_e\upsilon_v = \upsilon_u\upsilon_m$. This definition does not depend
    on the choice of $x$ because of the bijectiveness property of $e$, and we
    immediately have the commuting diagram
    \begin{center}
      \begin{tikzcd}
        X \arrow[d, kleisli, "u"'] \arrow[r, "e", kleisli, two heads] & Y_1
        \arrow[d, kleisli, "v"] \arrow[ld, kleisli, "\phi"'] \\
        Y_2 \arrow[r, "m", kleisli, tail] & Z
      \end{tikzcd}
    \end{center}

  \item $(\E_3, \M_3)$. Any $f: X \rightarrow M \times Y + 1$ factors
    through $e: X \rightarrow M \times Z + 1$ and $m: Z \rightarrow M \times Y
    + 1$ with
    \[ Z = \{ y \in Y \mid \exists x \in X, f(x) = (\--, y) \} \] $e(x) = f(x)$
    and $m(y) = (\varepsilon, y)$.

    Finally, given a commuting diagram
    \begin{center}
      \begin{tikzcd}
        X \arrow[d, kleisli, "u"'] \arrow[r, "e", kleisli, two heads] & Y_1
        \arrow[d, kleisli, "v"] \\
        Y_2 \arrow[r, "m", kleisli, tail] & Z
      \end{tikzcd}
    \end{center}
    the only possible choice for a $\phi: Y_1 \rightarrow M \times Y_2 + 1$ is
    given by $\phi(y_1) = \bot$ if $v(y_1) = \bot$ and $\phi(y_1) =
    (\upsilon_v\inv{\upsilon_m}, y_2)$ if $e(x) = (\upsilon_e, y_1)$, $u(x) =
    (\upsilon_u, y_2)$, $v(y_1) = (\upsilon_v, z)$, $m(y_2) = (\upsilon_m, z)$ and
    $\upsilon_e\upsilon_v = \upsilon_u\upsilon_m$. This definition does not depend
    on the choice of $x$ because $m$ is injective, and we immediately have
    the commuting diagram
    \begin{center}
      \begin{tikzcd}
        X \arrow[d, kleisli, "u"'] \arrow[r, "e", kleisli, two heads] & Y_1
        \arrow[d, kleisli, "v"] \arrow[ld, kleisli, "\phi"] \\
        Y_2 \arrow[r, "m", kleisli, tail] & Z
      \end{tikzcd}
    \end{center}
    \vspace*{-17pt}
    \qedhere
  \end{itemize}
\end{proof}
\noindent The factorization system we choose to define the minimal $M$-transducer is
$(\E_3,\M_3) = (\Surj, \Inj \cap \Inv \cap \Tot)$, because it generalizes the
factorization system that defines the minimal transducer (with output in a free
monoid). It will be our main factorization system, and as such from now on we
reserve the notation $(\E, \M)$ for it.

\Cref{thm:minimal_automaton} and
\Cref{proposition:minimal_object_infimum_divisibility} show that $(\E, \M)$
indeed gives rise to a useful notion of minimal transducer.

\begin{cor}
  \label{corollary:minimal_transducer}
  When $\Kl(\T_M)$ has all countable powers of $1$, the $(\E,\M)$-minimal
  transducer recognizing a $(\Kl(\T_M), 1, 1)$-language $\L$ is well-defined and
  has state-set $S^{min} = \{ \red(\inv{w}\L) \mid w \in A^* \} \cap (M +
  1)_*^{A^*}$, initial state $s_0^{min} = \red \L$, initialization value
  $\upsilon_0^{min} = \lgcd \L$, termination function $t^{min}(\Lambda) =
  \Lambda(e)$ and transition functions $\red(\inv{w}\L) \odot^{min} a =
  (\lgcd(\inv{(wa)}\L), \red(\inv{(wa)}\L))$. It is characterized by the
  property that all its states are reachable from the initial state and
  recognize distinct left-coprime functions.
\end{cor}

Why are $(\E_1,\M_1)$ and $(\E_2,\M_2)$ also interesting, then? They do not give
rise to useful notions of minimality, but they show that the computation of
$\Obs$ can be split into substeps. Indeed, since $\E_1 \subset \E_2 \subset
\E_3$ (and equivalently $\M_1 \supset \M_2 \supset \M_3$), $(\E_i, \M_i)_{1 \le
  i \le 3}$ is a \emph{quaternary factorization system}:

\begin{cor}
  For every $\Kl(\T_M)$-arrow $f: X \klarrow Y$, there is a unique (up to
  isomorphisms) factorization of $f$ into
  \begin{center}
    % https://tikzcd.yichuanshen.de/#N4Igdg9gJgpgziAXAbVABwnAlgFyxMJZABgBpiBdUkANwEMAbAVxiRAA0QBfU9TXfIRQBGclVqMWbAFoB9Yd14gM2PASIAmMdXrNWiEHI2K+qwUQDM2iXpmyLJ5fzVDkAFmu6pBgJrdxMFAA5vBEoABmAE4QALZIZCA4EEiiIAAWMHRQbDgA7hAZWQg6kvogrDwR0XGIqUlIWumZ2QZ5Bc3FiXRYDDndvSW2BuHyjlGxDdT1iFZNWTn5hVCdOP19PSCD3iAjxpU71Uiz0x5dG61r1Ax0AEYwDAAKzuYGDDDhOJs223FcFFxAA
    \begin{tikzcd}
      X \arrow[r, "f_1", e1, kleisli, two heads] & Z_1 \arrow[r, "f_2", m1, e2, kleisli, two heads, tail] & Z_2
      \arrow[r, "f_3", m2, kleisli, two heads, tail] & Z_3 \arrow[r, "f_4", kleisli, tail] & Y
    \end{tikzcd}
  \end{center}
  such that $f_1 \in \E_1$, $f_2 \in \E_2 \cap \M_1$, $f_3 \in \E_3 \cap \M_2$ and
  $f_4 \in \M_3$.
\end{cor}

Note how we respectively write $\klarrowin[->>]{e1}$, $\klarrowin[->>]{e2}$,
$\klarrowin[>->]{m1}$ and $\klarrowin[>->]{m2}$ for arrows in $\E_1$, $\E_2$,
$\M_1$ and $\M_2$ (but stick with $\kltwoheadarrow$ and $\klarrowtail$ for
arrows in $\E = \E_3$ and $\M = \M_3$).

Intuitively, this results means that the computation of any $f$ can be factored
into four parts: first forgetting some inputs ($f_1$ belongs to $\Surj \cap \Inj
\cap \Inv$ but need not belong to $\Tot$), then producing non-invertible
elements of the output monoid ($f_2$ belongs to $\Surj \cap \Inj \cap \Tot$ but
need not belong to $\Inv$), then merging some inputs together ($f_3$ belongs to
$\Surj \cap \Inv \cap \Tot$ but need not belong to $\Inj$) and finally embedding
the result into a bigger set ($f_4$ belongs to $\Inv \cap \Inj \cap \Tot$ but need
not belong to $\Surj$).

In particular, the $\E$-quotient $\Reach \A \kltwoheadarrow \Obs(\Reach \A)$
factors as follows.

\begin{defi}
  \label{def:minimization_steps}
  Given an $M$-transducer $\A$ recognizing the $(\Kl(\T_M), 1, 1)$-language
  $\L$, define $\Total \A$ and $\Prefix \A$ to be the $(\E_1,\M_1)$- and
  $(\E_2,\M_2)$-factorizations of the final arrow $\Reach \A \klarrow \A^{final}(\L)$:
  \begin{center}
    % https://tikzcd.yichuanshen.de/#N4Igdg9gJgpgziAXAbVABwnAlgFyxMJZABgBpiBdUkANwEMAbAVxiRAB12AlGOgYwAWAAk4BBEAF9S6TLnyEUARnJVajFm04AVCDkYj24qTOx4CRAEwrq9Zq0Qd2ABQBOMAGZYAHgaPSQGKbyRADM1mp2muwA8gBGcL6S-oFy5igALOG2Gg5iAHrAnmCMEgAUnAAyAJSSqjBQAObwRKDuLhAAtkhkIDgQSMogArxQbDgA7hDDdFAIxiBtnQPUfUhWQyNjk9OzICt0WAxjB0fzi12I66uIYRszW1MjCPuHx69n7Re315m9J29HagMOixGAMJyyMwKEAMDw4WoSIA
    \begin{tikzcd}
      \Reach \A \arrow[r, e1, kleisli, two heads] & \Total \A \arrow[r, e2, m1,
      kleisli, two heads, tail] & \Prefix \A \arrow[r, m2, kleisli, two heads,
      tail] & \Min \L \arrow[r, kleisli, tail] & \A^{final}(\L)
    \end{tikzcd}
  \end{center}
\end{defi}

In practice, if $\A = (S, (u_0, s_0), t, {\odot})$,
\begin{itemize}
\item $\Reach \A$ has state-set the set of states in $S$ that are reachable from
  $s_0$;
\item $\Total \A$ has state-set $S'$ the set of states in $S$ that recognize a
  function defined for at least one word (in particular if $\A$ recognizes
  $\bot^{A^*}$ then $(u_0, s_0)$ is set to $\bot$);
\item $\Prefix \A = (S', (u_0\lgcd(\L_{s_0}), s_0), t', {\odot'})$, where $\L_s$
  is the function recognized from a state $s \in S$ in $\A$, is obtained from
  $\Total \A$ by setting $t'(s) = \inv{\lgcd(\L_s)}t(s)$ and $s \odot' a =
  (\inv{\lgcd(\L_s)}(s \ocircle a)\lgcd(\L_{s \cdot a}), s \cdot a)$;
\item $\Min \L \cong \Obs(\Reach \A)$ is obtained from $\Prefix \A$ by merging two
  states $s_1$ and $s_2$ whenever they recognize functions that are equal up to
  invertibles on the left in $\Prefix \A$, that is when $\red(\L_{s_1}) =
  \red(\L_{s_2})$ in $\A$.
\end{itemize}
In particular, these four steps (computing $\Reach \A$, $\Total \A$, $\Prefix
\A$ and finally $\Obs \A$) match exactly the four steps into which all the
algorithms for minimizing (possibly monoidal) transducers are decomposed
\cite{breslauerSuffixTreeTree1998,choffrutMinimizingSubsequentialTransducers2003,eisnerSimplerMoreGeneral2003,gerdjikovGeneralClassMonoids2018,bealComputingPrefixAutomaton2000}.

\begin{exa}
  For the transducer $\B$ of \Cref{fig:monoidal_transducer} seen as a transducer
  with output in the free commutative monoid $\Sigma^{\oast}$, the corresponding
  minimal transducer $\Min \L$ is computed step-by-step in
  \Cref{fig:transducer_minimization}.

  Notice in particular how in \Cref{fig:total} both the functions recognized by
  the states $1$ and $3$ are left-divisible by $\alpha$ hence $\alpha$ is pulled
  back to the initialization value in \Cref{fig:prefix}. This would not have
  happened had $M$ been the free monoid $\Sigma^*$, and the corresponding
  minimal transducer would have been different.

  \begin{figure}[ht] \centering
    \captionsetup[subfigure]{justification=centering}
    \begin{subfigure}{0.4\textwidth}
      \centering \includefigure{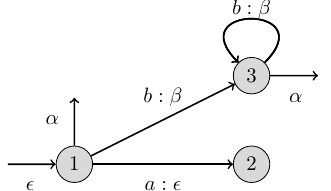}
      \caption{$\Reach \B$}
      \label{fig:reach}
    \end{subfigure}
    \medskip
    \begin{subfigure}{0.4\textwidth}
      \centering \includefigure{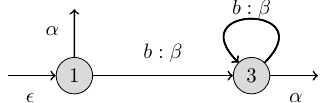}
      \caption{$\Total \B$}
      \label{fig:total}
    \end{subfigure}
    \\
    \begin{subfigure}{0.4\textwidth}
      \centering
      \includefigure{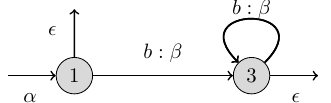}
      \caption{$\Prefix \B$}
      \label{fig:prefix}
    \end{subfigure}
    \medskip
    \begin{subfigure}{0.4\textwidth}
      \centering
      \includefigure{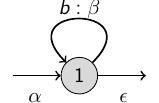}
      \caption{$\Obs(\Reach \B)$}
      \label{fig:obs_reach}
    \end{subfigure}
    \caption{Increasingly small transducers recognizing the same function as
      the transducer of \Cref{fig:monoidal_transducer} when $M = \Sigma^{\oast}$}
    \label{fig:transducer_minimization}
  \end{figure}
\end{exa}

\section{Active learning of minimal monoidal transducers}
\label{sec:learning}

Let $M$ be a monoid satisfying the conditions of
\Cref{lemma:countable_powers_sufficient_conditions} and consider now a function
$A^* \rightarrow M+1$ seen as a $(\Kl(\T_M),1,1)$-language $\L$.
\Cref{thm:minimal_automaton} tells us that the minimal $M$-transducer
recognizing $\L$ exists, is unique up to isomorphism and is given by
\Cref{corollary:minimal_transducer}, but does not tell us whether this minimal
transducer is computable. For this to hold we need that the product in $M$, the
left-gcd of two elements in $M$ --- written $\wedge$ --- and the function
\LeftDivide ~--- that takes as input $\delta, \upsilon \in M$ and outputs a
$\nu$ such that $\upsilon = \delta \nu$ or fails if there is none --- be all
computable, and that equality up to invertibles on the left be decidable (and
that the corresponding invertible be computable as well). We extend these
operations to $M+1$ by means of $u \bot = \bot u = \bot$, $u \wedge \bot = \bot
\wedge u = \bot$ and $\LeftDivide(\delta, \bot) = \bot$. For the computations to
terminate we additionally require that $\Min \L$ have finite state-set and $M$
be right-n\oe{}therian, so that $(\Min \L)(\St)$ is n\oe{}therian for the
factorization system $(\Surj,\Inj \cap \Inv \cap \Tot)$ of
\Cref{prop:factorization_systems}:

\begin{lem}
  \label{lemma:M-noetherianity}
  An object $X$ of $\Kl(\T_M)$ is $\M$-n\oe{}therian if and only if it is a finite
  set, in which case $\len_\M(m: Y \rightarrowtail X) = \card{X} - \card{Y}$.
\end{lem}
\begin{proof}
  Let $(x_i: 1 \klarrowtail X)_{i \in \N}$ be an infinite sequence of distinct
  elements of an infinite set $X$. Then the $m_n = [\kappa_i]_{i = 0}^n:
  \coprod_{i = 0}^n 1 \klarrowtail \coprod_{i = 0}^{n + 1} 1$ provide a
  counter-example to the $\M$-n\oe{}therianity of $X$ as none of them are
  isomorphisms (they are not surjective) yet $[x_i]_{i = 0}^{n+1} \circ m_n =
  [x_i]_{i = 0}^n$ and $[x_i]_{i = 0}^n$ is always an $\M$-morphism. Hence
  infinite sets are never $\M$-n\oe{}therian. If $X$ and the sequence $(x_i)_{i \in
  \N}$ were
  finite instead, this example would prove that $\len_\M(m: Y \rightarrowtail
  X) \ge \card{X} - \card{Y}$.

  Conversely, a strict chain of $\M$-subobjects of $X$ is a strict chain of
  subsets $X_0 \subsetneq X_1 \subsetneq \cdots$ of $X$. In particular, the
  cardinality of these subsets is strictly increasing: if $X$ is finite, the
  chain must be finite as well, and its length at most $\card{X} - \card{X_0}$.
\end{proof}

\begin{lem}
  \label{lemma:E-artinianity}
  An object $X$ of $\Kl(\T_M)$ is $\E^{op}$-n\oe{}therian if and only if it is a finite
  set and either $M$ is right-n\oe{}therian or $X = \varnothing$, in which case
  $\colen_\E(e: X \kltwoheadarrow Y) = \card{X} - \card{Y} + \rk e$ where $\rk
  e = \sum_{e(x) = (\upsilon, y)} \rk \upsilon$.
\end{lem}
\begin{proof}
  Let $(x_i)_{i \in \N}$ be an infinite sequence of distincts elements of an
  infinite set $X$, and set $X_n = \{ x_i \mid 0 \le i \le n \}$. Let $e_n: X
  \kltwoheadarrow X_n$ be defined by $e(x_i) = (\varepsilon, x_i)$ for $i \le n$
  and $e(x) = \bot$ otherwise, and let $e_n^{n+1}: X_{n+1} \kltwoheadarrow X_n$
  be the restriction of $e_n$ to $X_{n+1}$. None of the $e_n^{n+1}$ are
  isomorphisms (they are not total) yet $e_n^{n+1} \circ e_{n+1} = e_n$:
  infinite sets are never $\E^{op}$-n\oe{}therian.

  Assume now $X$ is not empty and $M$ is not right-n\oe{}therian: there is an
  element $x_* \in X$ and two sequences $(\upsilon_n)_{n \in \N}$ and
  $(\nu_n)_{n \in \N}$ of elements of $M$ such that for all $n \in \N$,
  $\upsilon_n \notin \invertibles{M}$ and $\nu_n = \nu_{n+1} \upsilon_n$. Let
  $e_n: X \kltwoheadarrow 1$ be defined by $e_n(x_*) = (\nu_n, *)$ and $e_n(x) =
  \bot$ for all other $x \in X$, and let $e_n^{n+1}: 1 \kltwoheadarrow 1$ be
  defined by $e_n^{n+1}(*) = (\upsilon_n, *)$. None of the $e_n^{n+1}$ are
  isomorphisms (they do not only produce invertible elements of $M$) yet
  $e_n^{n+1} \circ e_{n+1} = e_n$: non-empty sets are never $\E^{op}$-n\oe{}therian
  when $M$ is not right-n\oe{}therian.

  Conversely, if $X$ is empty then it is immediately $\E^{op}$-n\oe{}therian: there
  is only one $\E$-morphism out of $X$, $\id_X$. Suppose now that $X$ is finite
  and $M$ right-n\oe{}therian, and consider a cochain $e_n^{n+1}: X_{n+1}
  \kltwoheadarrow X_n$ of $\E$-quotients $e_n: X \kltwoheadarrow X_n$. Since $X$
  is finite, at most $\card{X} - \card{X_0}$ of the $\E$-quotients $X_{n+1}
  \kltwoheadarrow X_n$ witness a decrease of the cardinality from their domain
  to their codomain and are not in $\Inj \cap \Tot$. Fix now an $x \in X$ such
  that $e_0(x) \neq \bot$ and write $e_n(x) = (\nu_n, x_n)$ and
  $e_n^{n+1}(x_{n+1}) = (\upsilon_n, x_n)$ (this is well-defined because
  $e_n^{n+1} \circ e_{n+1} = e_n$). Then $\nu_n = \nu_{n+1}\upsilon_n$ for all
  $n \in \N$, hence since $M$ is right-n\oe{}therian only a finite number of the
  $e_n^{n+1}$, at most $\rk \nu_0$, produce a non-invertible element on
  $x_{n+1}$. This is true for all $x \in X$, hence a finite number of the
  $e_n^{n+1}$, at most $\rk e_0$, are not in $\Inv$, and a finite number of
  them, at most $\card{X} - \card{X_0} + \rk e_0$, are not in $\Iso$: $X$ is
  $\E^{op}$-n\oe{}therian and $\colen_\E(e_0: X \kltwoheadarrow X_0) \le \card{X} -
  \card{X_0} + \rk e_0$.

  Finally, if $\rk e$ is finite, in light of this proof it is now easy to build
  a strict cochain of $\E$-quotients of $X$ starting with $e: X \kltwoheadarrow
  Y$ that has length exactly $\card{X} - \card{Y} + \rk e$. For each $\nu(x) \in
  M$ such that $e(x) = (\nu(x), y)$ for some $x \in X$ and $y \in Y$, write
  indeed $\upsilon_1(x), \ldots, \upsilon_{\rk \nu(x)}(x)$ for a sequence of
  non-invertible divisors of $\nu$ of maximum length. Each morphism between two
  consecutive $\E$-quotients in the cochain should either decrease the size of
  the quotient by $1$, or produce exactly one of the $\upsilon_i(x)$ on $x$ for
  exactly one $x \in X$. Similarly, if $\rk e$ is infinite there sequences of
  divisors of some $\nu(x)$ or arbitrary length and it is then easy to build
  strict cochains of $\E$-quotients of $X$ of arbitrary lengths.
\end{proof}

The categorical framework of \Cref{sec:categorical_framework} can be extended
with an abstract minimization algorithm
\cite{aristoteFunctorialApproachMinimizing2023}. With the output category
described in \Cref{sec:category_monoidal_transducers}, an instance of this is in
particular Gerdjikov's algorithm for minimizing monoidal transducers
\cite{gerdjikovGeneralClassMonoids2018}, and even shows that the latter is still
valid under the conditions discussed in
\Cref{sec:category_monoidal_transducers:initial_final_objects} and terminates as
soon as $M$ is right-n\oe{}therian. However, we focus here on a second way to
compute the minimal transducer recognizing $\L$, namely learning it through
membership and equivalence queries, that is relying on a function $\Eval{\L}$
that outputs the value of $\L$ on input words and a function $\Equiv{\L}$ that
checks whether the hypothesis transducer is $\Min \L$ or outputs a
counterexample otherwise. Such an algorithm is an instance of the \textsc{FunL*}
algorithm described in \Cref{sec:categorical_framework:learning} and thus
terminates as soon as $(\Min \L)(\St)$ is $(\E,\M)$-n\oe{}therian. We now give a
practical description of this categorical algorithm: we explain how to keep
track of the minimal biautomaton and how to check whether $\varepsilon_{Q,T}^{min}$
is in $\E$ and $\M$. This is summarized by
\Cref{alg:learning_monoidal_transducers}
(\cpageref{alg:learning_monoidal_transducers}).

The algorithm for learning the minimal monoidal transducer recognizing $\L$ is
very similar to Vilar's algorithm (described in \Cref{sec:introduction}), the
main difference being that the longest common prefix is now the left-gcd and
that, in some places, testing for equality is now testing for equality up to
invertibles on the left. It maintains two sets $Q$ and $T$ that are respectively
prefix-closed and suffix-closed, and tables $\Lambda: Q \times (A \cup \{ e \})
\rightarrow M + 1$ and $R: Q \times (A \cup \{ e \}) \times T \rightarrow M +
1$. They satisfy that, for all $a \in A \cup \{ e\}$, $\Lambda(q,a) R(q,a,t) =
\L(\word{ q a t })$ and $R(q, a, \cdot)$ is left-coprime, hence $\Lambda(q, a)$
is a left-gcd of $(\L(\word{ q a t }))_{t \in T}$. The algorithm then extends
$Q$ and $T$ until some closure and consistency conditions are satisfied, and
builds a hypothesis transducer $\H(Q,T)$ using $\Lambda$ and $R$: its state-set
$S$ can be constructed by, starting with $e \in Q$, picking as many $q \in Q$
such that $R(q,e,\cdot)$ is not $\bot^T$ and such that, for every other $q' \in
S$, $R(q,e,\cdot)$ and $R(q',e,\cdot)$ are not equal up to invertibles on the
left; it then has initial state $e \in Q$, initialization value $\Lambda(e, e)$,
termination function $t = q \in S \mapsto R(q, e, e)$ and transition functions
given by $q \odot a = (\LeftDivide(\Lambda(q,e), \Lambda(q,a)) \chi, q')$ for
$q, q' \in S$ such that $R(q,a,\cdot) = \chi R(q',e,\cdot)$. The algorithm then
adds the counter-example given by $\Equiv{\L}(\H(Q,T))$ to $Q$ and builds a new
hypothesis automaton until no counter-example is returned and $\H(Q,T) = \Min
\L$.

Closure issues happen when $\varepsilon_{Q,T}^{min}$ is not in $\E = \Surj$, that
is when there is a $qa \in QA$ such that $R(q,a,\cdot) \neq \chi R(q',e,\cdot)$
for every other $q' \in Q$ and $\chi \in \invertibles{M}$, and in that case $qa$
should be added to $Q$. Consistency issues happen when the $\E$-factor of
$\varepsilon_{Q,T}^{min}$ is not in $\M = \Inj \cap \Inv \cap \Tot$, i.e., if it is
not in $\Tot$, in $\Tot$ but not in $\Inv \cap \Tot$, or in $\Inv \cap \Tot$ but
not in $\Inj \cap \Inv \cap \Tot$: the quaternary factorization system described
in \Cref{sec:category_monoidal_transducers:factorization_systems} thus also
explains the different kinds of consistency issues we may face. In practice,
there is hence a consistency issue if there is an $at \in AT$ such that
respectively: either there is a $q \in Q$ such that $R(q,a,t) \neq \bot$ but
$R(q,e,T) = \bot^T$; or there is a $q \in Q$ such that $\Lambda(q,e)$ does not
left-divide $\Lambda(q,a)R(q,a,t)$; or there are some $q,q' \in Q$ and $\chi \in
\invertibles{M}$ such that $R(q,e,T) = \chi R(q',e,T)$ but
$\LeftDivide(\Lambda(q,e), \Lambda(q,a)R(q,a,t)) \neq \chi
\LeftDivide(\Lambda(q',e), \Lambda(q',a)R(q',a,t))$. In each of these cases $at$
should be added to $T$.

If $\A$ is a monoidal transducer, write $\card{\A}_{\St} = \card{\A(\St)}$ and
$\rk(\A) = \sum_{s \in \A(\St)} \rk(\lgcd(\L_s))$ (where $\L_s$ is the partial
function recognized by $\A$ when $s \in \A(\St)$ is chosen to be the initial
state). The number of updates to $Q$ and $T$, hence in particular of calls to
$\Equiv{\L}$, is bounded linearly by $\card{\Min\L}_{\St}$ and $\rk(\Min \L)$
(although this latter quantity is not necessarily finite):

\begin{thm}
  \Cref{alg:learning_monoidal_transducers} is correct and terminates as soon
  as $\Min \L$ has finite state-set and $M$ is right-n\oe{}therian. It makes at
  most $3\card{\Min \L}_{\St} + \rk(\Min \L)$ updates to $Q$
  (\cref{alg:learning_monoidal_transducers:line:update-Q,alg:learning_monoidal_transducers:line:add-counterexample})
  and at most $\rk(\Min \L) + \card{\Min \L}_{\St}$ updates to $T$
  (\cref{alg:learning_monoidal_transducers:line:update-T}).
\end{thm}
\begin{proof}
  Notice first that \Cref{alg:learning_monoidal_transducers} is indeed the
  instance of \Cref{alg:FunL*} in $\Kl(\T_M)$: for all $q \in Q$ and $a \in A
  \cup \{ \varepsilon \}$, $\Lambda(q,a)$ is a left-gcd of $L(q, a, \cdot) =
  (\L(\word{qat}))_{t \in T} = \Lambda(q,a)R(q, a, \cdot)$, hence $R(q,a,\cdot)$
  is left-coprime and there is a $\chi \in \invertibles{M}$ such that
  $\Lambda(q, a)\inv{\chi} = \lgcd(L(q, a, \cdot))$ and $\chi R(q, a, \cdot) =
  \red(L(q, a, \cdot))$. Hence $\factorization{Q}{T} = \{ \red(L(q, e, \cdot))
  \mid q \in Q, L(q, e, \cdot) \neq \bot^T \}$ is the quotient of the set $\{
  R(q, e, \cdot) \mid q \in Q, R(q, e, \cdot) \neq \bot^T \}$ by equality up to
  invertibles on the left. It follows that $\factorization{Q}{T}
  \rightarrowtail\factorization{(Q \cup \{ qa \})}{ T}$ is an $\E$-morphism if
  and only if it is surjective, that is if and only if the condition on line 7
  is not satisfied, and $\factorization{Q}{(T \cup \{ at \})}
  \twoheadrightarrow\factorization{Q}{ T}$ is an $\M$-morphism if and only if it
  is total, produces only invertibles elements and is injective, that is if and
  only if it respectively does not satisfy any of the three conditions on line
  9.

  The correction and termination is then given by
  \Cref{thm:FunL*_correction_termination}, thanks to
  \Cref{lemma:M-noetherianity,lemma:E-artinianity}. These two lemmas also
  provide the complexity bound of the algorithm, as
  \Cref{thm:FunL*_correction_termination} is proven in
  \cite{colcombetLearningAutomataTransducers2020} by showing that each addition
  to $T$ contributes to a morphism in a strict chain of $\M$-subobjects of
  $(\Min \L)(\St)$ starting with $\factorization{\{\varepsilon\}}{A^*} \rightarrowtail
  \factorization{A^*}{A^*} = (\Min \L)(\St)$
  \cite[Lemma~33]{colcombetLearningAutomataTransducers2020}, and each addition
  to $Q$ contributes to a morphism in a chain of $\E$-quotients of $(\Min
  \L)(\St)$ ending with $(\Min \L)(\St) = \factorization{A^*}{A^*} \twoheadrightarrow
  \factorization{A^*}{\{\varepsilon\}}$
  \cite[Lemma~33]{colcombetLearningAutomataTransducers2020} and whose
  isomorphisms may only be contributed by the addition of a counter-example
  outputted by $\Equiv{\L}$ and are immediately followed by a non-isomorphism in
  the chain for $T$ or the cochain for $Q$
  \cite[Lemma~36]{colcombetLearningAutomataTransducers2020}.
\end{proof}

Our algorithm also differs from Vilar's original one in a small additional way:
the latter also keeps track of the left-gcds of every $\Lambda(q,\tilde{a})$
where $\tilde{a}$ ranges over $A \cup \{ e \}$ and $q \in Q$ is fixed, and
checks for consistency issues accordingly. This is a small optimization of the
algorithm that does not follow immediately from the categorical framework. In
\Cref{sec:introduction} we thus actually provided an example run of our version
of the algorithm when the output monoid is a free monoid. This also provides
example runs of our algorithm for non-free output monoids, as quotienting the
output monoid will only remove closure and consistency issues and make the run
simpler. For instance letting $\alpha$ commute with $\beta$ for the transducer
of \Cref{fig:learned_transducer} would have removed the closure issue and the
need to add $a$ to $Q$ while learning the corresponding monoidal transducer, and
letting $\alpha$ also commute with $\gamma$ would have removed the first
consistency issue to arise and the need to add $a$ to $T$.

\section{Summary and future work}
\label{sec:conclusion}

In this work, we instantiated Colcombet, Petri\c{s}an and Stabile's active
learning categorical framework with monoidal transducers. We gave some simple
sufficient conditions on the output monoid for the minimal transducer to exist
and be unique, which in particular extend Gerdjikov's conditions for
minimization to be possible \cite{gerdjikovGeneralClassMonoids2018}. Finally, we
described what the active learning algorithm of the categorical framework
instantiated to in practice under these conditions, relying in particular on the
quaternary factorization system in the output category.

This work was mainly a theoretical excursion and was not motivated by practical
examples where monoidal transducers are used. One particular application that
could be further explored is the use of transducers with outputs in trace
monoids (and their learning) to programatically schedule jobs, as mentioned in
the introduction. We also leave the search for other interesting examples for
future work.

Some intermediate results of this work go beyond what the categorical framework
currently provides and could be generalized. The use of a quaternary
factorization system (or any $n$-ary factorization system) would split the
algorithms into several substeps that should be easier to work with. Here our
factorization systems seemed to arise as the image of the factorization system
on $\Set$ through the monad $\T_M$; generalizing this to other monads could
provide meaningful examples of factorization systems in any Kleisli category.
Finally, we mentioned in
\Cref{sec:category_monoidal_transducers:initial_final_objects} that a problem
with the current framework is that it may only account for the minimization of
both finite and infinite transition systems at the same time, and conjectured
that we could restrict to only the finite case by working in a subcategory of
well-behaved transducers: this subcategory is perhaps an instance of a general
construction that has its own version of \Cref{thm:minimal_automaton}, so as to
still have a generic way to build the initial, final and minimal objects.

\section*{Acknowledgments}
The author is grateful to Daniela Petri\c{s}an and Thomas Colcombet for
  fruitful discussions.

\bibliographystyle{alphaurl}
\bibliography{bibtex.bib}

\appendix  % I moved the algorithm to the appendix. If you do not agree with this, let us know
\section{}

\begin{algorithm}
  \caption{The \textsc{FunL*}-algorithm for monoidal transducers}
  \label{alg:learning_monoidal_transducers}
  \begin{algorithmic}[1] \REQUIRE $\Eval{\L}$ and
    $\Equiv{\L}$ \ENSURE $\Min_M(\L)$

    \STATE $Q = T = \{ e \}$

    \FOR{$a \in A \cup \{ e \}$}

    \STATE $\Lambda(e, a) = \Eval{\L}(a)$ \STATE $R(e, a, e) = \varepsilon$

    \ENDFOR

    \LOOP

    \IF{there is a $qa \in QA$ such that $\forall q' \in Q, \chi \in
      \invertibles{M}, R(q, a, \cdot) \neq \chi R(q', e, \cdot)$}

    \STATE add $qa$ to $Q$ \label{alg:learning_monoidal_transducers:line:update-Q}

    \ELSIF{there is an $at \in AT$ such that
      \begin{itemize}
      \item \textbf{either} there is a $q \in Q$ such that $R(q, a, t) \neq
        \bot$ but $R(q, e, T) = \bot^T$;
      \item \OR there is a $q \in Q$ such that $\Lambda(q, e)$ does not
        left-divide $\Lambda(q, a)R(q, a, t)$;
      \item \OR there are $q, q' \in Q$ and $\chi \in \invertibles{M}$ such that
        $R(q, e, T) = \chi R(q', e, T)$ but $\LeftDivide(\Lambda(q, e),
        \Lambda(q, a) R(q, a, t)) \neq \chi \LeftDivide(\Lambda(q',e),
        \Lambda(q', a) R(q', a, t))$
      \end{itemize}}

    \STATE add $at$ to $T$ \label{alg:learning_monoidal_transducers:line:update-T}

    \ELSE \STATE build $\H(Q,T) = (S,(\upsilon_0,s_0),t,{\odot})$ given by:
    \begin{itemize}
    \item $S \subseteq Q$ is built by starting with $e \in S$ and adding as
      many $q \in Q$ as long as $R(q,e,\cdot) \neq \bot$ and $\forall q' \in
      S, \chi \in \invertibles{M}, R(q,e,\cdot) \neq \chi R(q',e, \cdot)$;
    \item $(\upsilon_0, s_0) = (\Lambda(e),e)$
    \item $q \odot a = (\LeftDivide(\Lambda(q,e), \Lambda(q,a)) \chi, q')$
      with $q \in S$, $\chi \in \invertibles{M}$ given by $R(q,a,\cdot) = \chi
      R(q',e,\cdot)$
    \item $t(q) = R(q,e,e)$
    \end{itemize}

    \IF{$\Equiv{\L}(\mathcal{H}_{Q,T}(\L))$ outputs some
      counter-example $w$}

    \STATE add $w$ and its prefixes to
    $Q$ \label{alg:learning_monoidal_transducers:line:add-counterexample}

    \ELSE

    \RETURN $\H(Q,T)$

    \ENDIF

    \ENDIF

    \STATE update $\Lambda$ and $R$ using $\Eval{\L}$

    \ENDLOOP
  \end{algorithmic}
\end{algorithm}

\end{document}